\documentclass[]{amsart}

\usepackage[longnamesfirst,sort]{natbib}
\bibpunct[, ]{(}{)}{;}{a}{,}{,}

\usepackage{amssymb,color,graphicx, fancyhdr}
\usepackage[table]{xcolor}

\usepackage[italian,english]{babel}
\usepackage[applemac]{inputenc}
\usepackage[T1]{fontenc} 
\usepackage{times}

\usepackage{mathrsfs}
\usepackage{enumitem}
\usepackage{multirow}

\usepackage{times}
\usepackage{helvet}
\usepackage{courier}
\usepackage{amsmath}
\usepackage{graphics}
\usepackage{graphicx}
\usepackage[misc,geometry]{ifsym}

\usepackage{cmll}
\usepackage{times}
\usepackage{helvet}
\usepackage{courier}
\usepackage{amssymb}
\usepackage{amsmath}
\usepackage{balance}
\usepackage{bussproofs}
\usepackage{graphics}
\usepackage{mathrsfs}
\usepackage[misc,geometry]{ifsym}
\usepackage{url}
\usepackage{xspace}

\newtheorem{theorem}{Theorem}
\newtheorem{lemma}[theorem]{Lemma}
\newtheorem{definition}[theorem]{Definition}
\newtheorem{proposition}[theorem]{Proposition}
\newtheorem{property}[theorem]{Property}
\newtheorem{corollary}[theorem]{Corollary}

\newtheorem*{example}{Example}

\newcommand\AC[1]{\AxiomC{$#1$}}
\newcommand\rl[1]{\RightLabel{#1}}
\newcommand\BC[1]{\BinaryInfC{$#1$}}
\newcommand\UC[1]{\UnaryInfC{$#1$}}
\newcommand\dip{\DisplayProof}
\renewcommand\t{\vdash}

\renewcommand{\L}{\ensuremath{\mathsf{L}}\xspace}
\newcommand{\IL}{\ensuremath{\mathsf{IL}}\xspace}
\newcommand{\CL}{\ensuremath{\mathsf{CL}}\xspace}
\newcommand{\ALL}{\ensuremath{\mathsf{ALL}}\xspace}
\newcommand{\MLL}{\ensuremath{\mathsf{MLL}}\xspace}
\newcommand{\LL}{\ensuremath{\mathsf{LL}}\xspace}
\newcommand{\MALL}{\ensuremath{\mathsf{MALL}}\xspace}
\newcommand{\HMALL}{\ensuremath{\mathsf{HMALL}}\xspace}

\newcommand{\IMALL}{\ensuremath{\mathsf{IMALL}}\xspace}
\newcommand{\R}{\ensuremath{\mathsf{R}}\xspace}
\newcommand{\AR}{\ensuremath{\mathsf{AR}}\xspace}

\newcommand{\BI}{\ensuremath{\mathsf{BI}}\xspace}

\begin{document}


\title{Judgment Aggregation in Non-Classical Logics}

\author{Daniele Porello}

\begin{abstract}
This work contributes to the theory of judgment aggregation by discussing a number of significant non-classical logics. After adapting the standard framework of judgment aggregation to cope with non-classical logics, we discuss in particular results for the case of Intuitionistic Logic, the Lambek calculus, Linear Logic and Relevant Logics. 
The motivation for studying judgment aggregation in non-classical logics is that they offer a number of modelling choices to represent agents' reasoning in aggregation problems. By studying judgment aggregation in logics that are weaker than classical logic, we investigate whether some well-known impossibility results, that were tailored for classical logic, still apply to those weak systems. 
\end{abstract}

\maketitle

\medskip

\normalsize

\section{Introduction}

Judgement aggregation (JA) has recently become a significant topic at the intersection of themes in social choice theory, multiagent systems, and knowledge representation. The reason is that JA provides a general theory for studying the procedures to aggregate agents' possibly heterogeneous attitudes into a collective attitude that reflects, as close as possible, the individual views. The original applications of judgment aggregation were related to voting theory and to the modellisation of the decision-making processes of collegial courts, cf. \cite{ListPettit2002,KornhauserSagerCLR1993}. In AI and multiagent systems, JA  provides a sound methodology to define ascriptions of propositional contents to collective entities; for that reason, JA has been related for instance to belief merging, cf \cite{pigozzi2006belief,konieczny2011logic}, and to ontology merging, cf \cite{PorelloEndrissJLC2014}.\par
As usual in formal modellings of agency, a modicum of rationality is presupposed in the understanding of agents, \cite{ListPettit2011}. Since JA is about the aggregation of logically connected propositions, the notion of individual rationality of JA is essentially related to the concept of logical consistency. The concept of collective rationality is therefore intended as the preservation of individual consistency, via the aggregation procedure, at the collective level.  As it is known from the beginning of the study of JA, the aggregation of individual consistent judgments may lead to inconsistent outcomes, cf. \cite{ListPettit2002}. For that reason, a careful investigation of the properties of the aggregation procedures and of the extent to which they affect the preservation of consistency is required. This line of study has been developed into a quite sophisticated theory of aggregation procedures. A number of introductory readings to judgment aggregation is now available, for instance, \cite{grossi2014judgment}, \cite{list2010introduction}, \cite{ListPuppe2009}, and \cite{EndrissHBCOMSOC2016}.\par
Important aggregation procedures, with the significant example of the majority rule, do not preserve individual consistency once agents are allowed to express their judgments on any possible proposition, that is, once the agenda on which the agents express their views is allowed to be any subset of the logical language. Hence, one can roughly divide the JA approach to cope with collective inconsistency into two directions: the first is the investigation of the procedures that are indeed capable of preserving consistency; the second is the characterisations of the agendas that guarantee consistent collective outcomes.\par In the standard view of JA, the logic that is used to model rationality is classical logic. Thus, the failure of preserving consistency that the results in judgment aggregation show concerns the notion of consistency defined by classical logic.\par
In this paper, I propose to extend the theory of judgment aggregation to a number of significant of non-classical logics.  By non-classical logic, we shall mean here logical systems that reject one or more principles of classical logic and provide an alternative view of reasoning. A guiding example is intuitionistic logic, see \cite{Dummett2000}: by rejecting for instance the principle of the excluded middle, intuitionism provides a constructive account of reasoning. Many logics for JA have been investigated, designed, and discussed; however, they usually extend classical propositional logics, instead of investigating systems that are alternative to or weaker than classical logic. A few significant exceptions shall be discussed later on.\par
The motivation for studying judgment aggregation in non-classical logic are essentially three. Firstly, there is a theoretical interest in extending the theory of judgment aggregation to logics that, from a mathematical point of view, significantly differ from classical logic. Secondly, by studying non-classical logics, we shall approach weak inference systems; thus, it is worth investigating whether, by weakening the logic that models the concept of rationality, standard impossibility results still hold for weaker system, or whether we can actually circumvent collective inconsistencies by weakening classical logic. Finally, as usual in justifying the investigation of non-classical logics, the adequacy of classical connectives to model reasoning may be questioned. It is interesting to notice that arguments against the material conditional of classical logic have also been suggested within the literature on JA. In particular, by \cite{dietrich2010possibility}, where classical implication is replaced by subjunctive implications of Lewis conditional logic. By studying JA in non-classical logic, we enable the choice among a number of definitions of logical connectives that may be appropriate for certain aggregation problems. For this reason, we shall focus in this paper on a number of well-established non-classical logics with significant impact in philosophical logic or in computer science. For instance, intuitionistic logic provides a constructive view of reasoning that may suit an evidence-based view of inferences, while relevant logic allows for defining a conditional that better copes with the paradoxes of material implication, cf. \cite{Anderson1992entailment}. 
We shall start discussing JA in non-classical logic by dealing with intuitionistic logic. Then, we shall develop our analysis of non-classical logics by means of substructural logics, cf. \cite{Paoli2002,Restall2002}. \emph{Substructural logics} are a family of logics weaker than classical logic that reject one or more of the core principle of classical reasoning; for instance, the monotonicity of the entailment or the commutativity of logical connectives.\par
The conceptual motivations for studying substructural logics are usually related to the idea of capturing a form of reasoning that better copes with the actual inferential practice of human or artificial agents. For instance, the weakest substructural logic that we discuss in this paper is the \emph{Lambek Calculus}, developed by \cite{Lambek1958}, for which connectives are non-commutative. This entails that the ordering of the formulas is crucial for reasoning. The order-dependency of inferences can be used to model aggregation problems with temporal dependencies among the issues of the agenda. For instance, suppose $A$ and $B$ in this order have been accepted, then the constraint $B \wedge A \rightarrow C$ does not apply in this situation to infer $C$.\par
Moreover, we approach \emph{linear logic}, introduced by \cite{Girard1987}, that captures a form of resource-bounded reasoning. For example, suppose the proper axiom $E \rightarrow C$ represents the inference ``if I have one euro ($E$), then I buy one coffee ($C$)''. In classical logic, one can infer, by means of the contraction principle, that $E \rightarrow E \wedge C$, namely, that I still have one euro, besides having the coffee. By dropping contraction, linear logic captures a the resource-sensitive aspect of causality: the antecedent has to be \textit{consumed} during the inferential process so that the consequent may hold, cf. \cite{Girard1995}.\par 
In philosophical logic, an important debate on the nature of logical implication emerged in the tradition of \emph{relevant logics}: \cite{Anderson1992entailment},  \cite{DunnRestall2002}. The family of relevant logics rejects, in particular, the monotonicity of the entailment and design therefore logics for which true implications exhibit the relevance of the antecedent of the conditional to the consequent. In particular, relevant logics reject the axiom $A \rightarrow (B \rightarrow A)$ that means that whenever $A$ holds, $A$ can be implied by any $B$, regardless of the relevance of $B$ for assessing $A$. Relevant logics have also been motivated as logics for modelling epistemic agents and as logics that model inferences that depend on the amount of information available to cognitive agents, \cite{Mares2004,AlloMaresErkenntnis2012,MasoloPorelloAIC2015}.\par
We will also briefly discuss fuzzy logics that are defined within the substructural realm. We shall not discuss in details the case of paraconsistent logics because we are interested here in studying the preservation of consistency via aggregation procedures. However, paraconsistency is approached in this paper  by presenting logics for which the principle of \emph{ex falso quodlibet} does not hold (e.g. in the case of linear and relevant logics).\par 
The methodology of this paper is proof-theoretical. This is motivated by the fact that we are going to introduce a number of logics with significantly different algebraic counterparts: the proof-theoretical methodology permits a compact treatment. Moreover, as we shall see, the proof-theoretical analysis allows for pinpointing the inferences principles that are responsible of the failures of preserving consistency in judgment aggregation.\par
The main contribution of this paper consists in the extension of the theory of judgment aggregation to the case of a number of non-classical logics. The choice of the logical systems in this paper is also motivated by the fact that, by dropping monotonicity and other principles of classical logic, it is possible to circumvent the collective inconsistency that threaten the judgment aggregation based on classical logics. Moreover, results for general monotonic logics have already been presented for instance by \cite{Dietrich2007} and for nonmonotonic logics have been recently presented in \cite{WenJANMSYnthese2017}. Another crucial aspect for preserving consistency that motivates the focus on linear and relevant logics is the distinction that this systems provide between additive and multiplicative, or extensional and intensional, logical connectives. As we shall see in Section \ref{section:mall} and in Section \ref{sec:jar}, the combination of the lack of monotonicity and the distinction between types of connectives allows for establishing positive results for the majority rule, that is, performing collective reasoning in those weak systems guarantees consistency.
By relying on that, in Section 9.1 a strategy for circumventing judgment aggregation paradoxes is proposed. The idea is to assess individual reasoning with respect to classical logic, as usual in judgment aggregation, while assessing collective reasoning with respect to a weak logic for which consistency is ensured. To enable the assessment of individual and collective reasoning with respect to possibly distinct logics, we shall slightly rephrase the standard framework of judgment aggregation.\par
The remainder of this paper is organised as follows. Section \ref{sec:background} introduces the background on substructural logics. In particular, we recall the sequent calculi for classical and intuitionistic logic, for the Lambek calculus, and for linear logic. For the case of relevant logics, as we shall see, the sequent presentation is not satisfactory, thus we shall approach relevant logics by means of Hilbert systems. Section \ref{sec:modelJA} introduces a model of judgment aggregation that is general enough to treat the case of the non-classical logics that we discuss here. In particular, we stress the interesting differences that the lacking of monotonicity, contraction, and commutativity entail on the model. Section \ref{sec:JAclassical} rephrases the standard results of JA in classical logic in the setting of this paper (i.e. in proof-theoretical terms) and discusses the case of conservative extensions of classical logic. Section \ref{sec:il} approaches JA for intuitionistic logic. Section \ref{sec:lambek} discusses JA for the Lambek calculus. Section \ref{section:mall} presents the case of linear logic. In particular, as we shall see, by focusing on a fragment of linear logic, general possibility results for judgment aggregation are achievable. Section \ref{sec:jar} extends the previous analysis to the case of relevant logics. Section \ref{sec:extensions} discusses a number of extensions of the previous treatment to other logics and discusses how to circumvent the impossibility results based on agendas in classical logic by rephrasing classical collective rationality with substructural reasoning. Section \ref{sec:relatedwork} surveys related work, by focussing in particular on the use of non-classical logics in JA, on the use of proof-theoretical methods, and on the relationship with the algebraic approach to the study of general logics of JA. Section \ref{sec:conclusion} concludes.   
 
\section{Background on substructural logics}\label{sec:background}

Sequent calculi were introduced by Gerhard Gentzen to study proofs in classical and intuitionistic logic. They provide an important theory in logic that allows for investigating the operational meaning of logical connectives, cf. \cite{gentzen1935untersuchungen}.\footnote{For an introduction to sequent calculi, see \cite{TroelstraBPT} and \cite{negri2008structural}.} Besides providing a fine-grained tool to analyse inferences, sequent calculi can be used to model a number of logics in a uniform and elegant way. A \emph{sequent} is an expressions of the form $\Gamma \vdash \Delta$, where $\Gamma$, the premises of the sequent, and $\Delta$, the conclusions of the sequent, are made out of formulas in a given logic.  The structure of $\Gamma$ and $\Delta$, as we shall see, depends on the logic; for instance, in what follows, they may be sets, multisets, or lists of formulas.\par
The intuitive reading of a sequent expression is: ``the conjunction of the formulas in $\Gamma$ entails the disjunction of the formulas in $\Delta$''.
A sequent calculus is specified by two classes of rules. The \emph{structural rules} determine the structure of the sequent and define how to handle the hypotheses in a proof, the assumptions of reasoning; for instance, they entail that in classical logic $\Gamma$ and $\Delta$ are \emph{sets} of formulas. The \emph{logical rules} define the operational meaning of the logical connectives.
A fundamental insight concerning the meaning of the logical rules is due to the tradition of substructural logics and in particular to Jean Yves Girard (cf. \cite{Girard1987}): The structural rules determine the behaviour of the logical connectives.  The meaning of the structural rule for reasoning is the following: \emph{weakening} (W) corresponds to the monotonicity of the consequence relation, \emph{contraction} (C) amounts to assuming that multiple occurrences of the same formula can be identified, and \emph{exchange} (E) forces the commutativity of conjunction and disjunction. When rejecting one or more structural rules, we simply cannot view the class of propositions about which reasoning is performed as a \emph{set}; therefore, in the subsequent presentation, we shall define the structure of $\Gamma$ and $\Delta$ as either a multiset or a list.\par
We introduce the language of classical propositional logic \CL as usual. Assume a set of propositional variable $\textbf{P}$, the set of formulas of \CL is defined as follows.

\begin{equation}
\mathcal{L}_{CL} := p \in \textbf{P} \mid \neg A \mid A \wedge A \mid A \vee A \mid A \rightarrow A
\end{equation}

The sequent calculus for \CL is presented in Table \ref{SequentClassical}. The logical rules may be presented either in \emph{multiplicative} form or in \emph{additive} form. In the former case, the premises of the rule have possibly different contexts that are combined in the conclusion, in the latter, the rules presuppose that the premises of the sequent share the same context.\footnote{The terms additive and multiplicative have been introduced by \cite{Girard1987}. The distinction corresponds to the division between \emph{extensional} and \emph{intensional} connectives in relevant logics, cf. \cite{Paoli2002}.}
The dependency of logical rules on structural rules is evident by noticing that classical logic is imposed as soon as we assume (left and right versions of) weakening (W L, and W R), contraction (C R and C L), and exchange (E L and ER). In that case, the additive and the multiplicative versions of the rules of connectives are provably equivalent. 

\begin{table}[h]
\begin{center}
\begin{center}

\textit{Identities}

\begin{tabular}[]{cc}
& \\
\AxiomC{} \RightLabel{\footnotesize ax} \UnaryInfC{$A \vdash A$} \DisplayProof  
& \AxiomC{$\Gamma, A \vdash \Delta$} \AxiomC{$\Gamma' \vdash A, \Delta'$} \RightLabel{\footnotesize cut}
\BinaryInfC{$\Gamma,
\Gamma' \vdash \Delta, \Delta'$} \DisplayProof
\\
 & \\
\end{tabular}

\textit{Structural Rules}

\begin{tabular}[]{cc}
&  \\
\AC{\Gamma, A, B, \Gamma' \t \Delta} \rl{E L} \UC{\Gamma, B, A, \Gamma' \t \Delta}
\DisplayProof & \AC{\Gamma \t \Delta, A, B, \Delta'} \rl{E R} \UC{\Gamma \t \Delta, B, A, \Delta'}
\DisplayProof\\
&  \\

\AC{\Gamma, A, A, \t \Delta} \rl{C L} \UC{\Gamma, A \t \Delta} \dip & \AC{\Gamma \t \Delta, A, A} \rl{C R} \UC{\Gamma \t \Delta, A} \dip\\
&  \\
\AC{\Gamma \t \Delta} \rl{W L} \UC{\Gamma, A \t \Delta} \dip  & \AC{\Gamma \t \Delta} \rl{W R} \UC{\Gamma  \t \Delta, A} \dip\\
& \\
\end{tabular}

\textit{Negation}

\begin{tabular}[]{cc}
& \\
\AC{\Gamma \t A, \Delta} \rl{$\neg$ L} \UC{\Gamma, \neg A \t \Delta}
\dip   &
\AC{\Gamma, A\t \Delta} \rl{$\neg$ R} \UC{\Gamma \t \neg A, \Delta}
\dip
\\
 & \\
\end{tabular}

\textit{Multiplicative presentation of logical connectives}

\begin{tabular}{ll}
&  \\
\AC{\Gamma \vdash A, \Delta} \AC{\Gamma' \vdash B, \Delta'}
 \LeftLabel{$\wedge$R} \BC{\Gamma, \Gamma' \vdash A\wedge B, \Delta, \Delta'} \dip  & 
 \AC{\Gamma, A, B \vdash \Delta} \rl{$\wedge$L} \UC{\Gamma, A \wedge B \vdash \Delta}
\dip      \\
& \\
\AC{\Gamma, A \vdash \Delta} \AC{\Gamma', B \vdash \Delta'} \LeftLabel{$\vee $L}
\BC{\Gamma, \Gamma', A \vee B \vdash \Delta, \Delta'} \dip  &

\AC{\Gamma \vdash A, B, \Delta} \rl{$\vee$ R}
\UC{\Gamma \vdash A \vee B, \Delta}   \dip \\
& \\
\AxiomC{$\Gamma, A \t \Delta$} \AxiomC{$\Gamma', B \t \Delta'$} \LeftLabel{$\rightarrow$L}
\BinaryInfC{$\Gamma, \Gamma' A \rightarrow B \t \Delta, \Delta'$} \DisplayProof & 
 \AxiomC{$\Gamma \t B, \Delta$}
\RightLabel{$\rightarrow$R}
\UnaryInfC{$\Gamma \t A \rightarrow B,\Delta$} \DisplayProof\\
& \\
\end{tabular}

\textit{Additive presentation of logical connectives}

\begin{tabular}[]{ccc}
&  \\
\AxiomC{$\Gamma \t A, \Delta$} \AxiomC{$\Gamma \t B, \Delta $} \LeftLabel{$\wedge$R}
\BinaryInfC{$\Gamma \t A \wedge B, \Delta$} \DisplayProof &
 \AxiomC{$\Gamma, A \t \Delta$} \RightLabel{$\with$L} \UnaryInfC{$\Gamma, A \wedge
B \t \Delta$} \DisplayProof
&
\AxiomC{$\Gamma, B \t \Delta$} \RightLabel{$\with$L} \UnaryInfC{$\Gamma, A \wedge
B \t \Delta$} \DisplayProof\\
& \\

\AxiomC{$\Gamma, A \t \Delta$} \AxiomC{$\Gamma, B \t \Delta$} \LeftLabel{$\vee$L}
\BinaryInfC{$\Gamma, A\vee B \t \Delta$} \DisplayProof & 
\AxiomC{$\Gamma \t A, \Delta$}
\RightLabel{$\vee$R}
\UnaryInfC{$\Gamma \t A \vee B,\Delta$} \DisplayProof
& 
\AxiomC{$\Gamma \t B, \Delta$}
\RightLabel{$\vee$R}
\UnaryInfC{$\Gamma \t A \vee B,\Delta$} \DisplayProof\\
& \\

\AxiomC{$\Gamma \t A, \Delta$} \AxiomC{$\Gamma, B \t \Delta$} \LeftLabel{$\rightarrow$L}
\BinaryInfC{$\Gamma, A \rightarrow B \t \Delta$} \DisplayProof & 
 \AxiomC{$\Gamma \t B, \Delta$}
\RightLabel{$\rightarrow$R}
\UnaryInfC{$\Gamma \t A \rightarrow B,\Delta$} \DisplayProof

&

 \AxiomC{$\Gamma \t A, \Delta$}
\RightLabel{$\rightarrow$R}
\UnaryInfC{$\Gamma \t A \rightarrow B,\Delta$} \DisplayProof
\\
& \\
\end{tabular}

\end{center}
\end{center}
\caption{Sequent Calculus for Classical Logic (CL)}
\label{SequentClassical}
\end{table}
\noindent

By disabling one or more structural rules, we enter the realm of  \emph{substructural logics}; for introductory readings cf. \cite{Paoli2002,Restall2002,Troelstra1992}.\par Intuitionistic logic (\IL) can also be construed as a substructural logic. Gentzen surprisingly showed that intuitionistic logic can be captured by simply imposing upon the sequent calculus for \CL that the right-hand side of the sequent contains at most one formula. This restriction is sufficient to prevent non-constructive principles, such as the excluded middle, to be provable. 
This restriction has been interpreted as a manner to disable structural rules locally by \cite{GirardEtAl1989}: in intuitionistic logic the structural rules are prevented on the right-hand side of the sequent and fully permitted on the left-hand side. 
Define the language of intuitionistic logic  ($\mathcal{L}_{\IL}$) as follows. Let $p \in \textbf{P}$ and $\bot$ the constant for false:

\begin{equation}
\mathcal{L}_{IL} = p \mid \bot \mid A \wedge A \mid A \vee A \mid A \rightarrow A
\end{equation} 

The sequent calculus of intuitionistic logic (\IL) is given by restricting the classical calculus to single-conclusion sequents. Negation is defined, as usual in intuitionistic logic, by means of $\bot$ and $\rightarrow$: $\neg A = A \rightarrow \bot$.

\subsection{Sequent Calculi for Substructural Logics}

The weakest substructural logic that we discuss here is obtained by removing all structural rules. The resulting calculus is known as the \emph{Lambek calculus} \L, developed by \cite{Lambek1958} to capture syntactic parsing of natural language sentences along the tradition of categorial grammars.
By rejecting exchange, the conjunction of \L, denoted by $\odot$, is non-commutative. Moreover, there are two order-sensitive implications $\setminus$ and $/$. Note that \L is an intuitionistic logic, in the sense that at most one formula appears in the right-hand side of the sequent. 
 We assume a minor extension of the Lambek calculus with the constant for false $\bot$ in order to define a form of intuitionistic negation in terms of the order sensitive implication $A \setminus \bot$ (cf. Table \ref{SequentLambek}) .\footnote{In fact, one can show that there are two negations definable in this way that depend on the order-sensitive implications: $A \setminus \bot$ and $\bot / A$. For our purposes, it is not worthy entering the details of the treatment of negation in $\L$. We refer to \cite{Wansing2007}, \cite{Restall2006} for strong negations in \L and to \cite{Abrusci1990} for the definition of the two negations for Lambek calculus.}

\begin{equation}
\mathcal{L}_{L} := p \in \textbf{P} \mid \bot \mid A \odot A \mid A \setminus A \mid A / A 
\end{equation}

\begin{table}[h]
\begin{center}

\begin{tabular}[]{cc}
& \\
\AxiomC{$\Gamma[A; B] \vdash C$} \RightLabel{$\odot$L} \UnaryInfC{$\Gamma[A \odot B] \vdash
C$} \DisplayProof  &  \AxiomC{$\Gamma \vdash A $} \AxiomC{$\Gamma' \vdash B$}
 \RightLabel{$\odot$R} \BinaryInfC{$\Gamma; \Gamma' \vdash A\odot B$} \DisplayProof      \\
& \\
\AxiomC{$\Gamma \vdash A$} \AxiomC{$\Delta[B] \vdash C$} \RightLabel{$\setminus$ L}
\BinaryInfC{$\Delta[\Gamma; A\setminus B] \vdash C$} \DisplayProof  &
\AxiomC{$A; \Gamma \vdash B$} \RightLabel{$\setminus$R}
\UnaryInfC{$\Gamma \vdash A \setminus B$}   \DisplayProof \\
& \\
\AxiomC{$\Gamma \vdash A$} \AxiomC{$\Delta[B] \vdash C$} \RightLabel{$/$ L}
\BinaryInfC{$\Delta[ B / A; \Gamma]  \vdash C$} \DisplayProof  &
\AxiomC{$\Gamma; A \vdash B$} \RightLabel{$/$R}
\UnaryInfC{$\Gamma \vdash A / B$}   \DisplayProof \\
\end{tabular}
\end{center}
\caption{Sequent calculus for the Lambek Calculus}
\label{SequentLambek}
\end{table}

By adding exchange to the Lambek calculus, we obtain linear logic  (\LL). That is, linear logic rejects the global validity of (W) and (C) both on the left and on the right hand side of the sequent. In \LL, sequents are multisets of formulas.  In Table \ref{SequentClassical}, we presented two ways of defining logical rules: an \emph{additive} version and a \emph{multiplicative} version. The two formulations are redundant in classical logic because of (W) and (C), that are sufficient to prove that the additive formulation and the multiplicative formulation are equivalent. If we drop weakening and contraction, additives and multiplicatives are no longer equivalent, hence we have to account for two different types of conjunctions and disjunctions with distinct operational meanings.  This operators are not visible in classical logic, because of the structural rules. Accordingly, in \LL there are two different types of conjunction, a multiplicative conjunction $\otimes$ (tensor) and an additive conjunction $\with$ (with), and two types of disjunctions, multiplicative $\parr$ (parallel) and additive $\oplus$ (plus). Implications can be defined by means of disjunctions and negations as usual, in \LL  the multiplicative implication is $A \multimap B \equiv \neg A \parr B$ and the additive implication $A \leadsto B \equiv \neg A \oplus B$.\footnote{For the additive implication, whose status as an implication is in fact questionable, we refer to \cite{TroelstraBPT}. The reason why the additive implication is not satisfactory is that, in a categorical jargon, the adjunction w.r.t. additive conjunctions does not hold, that is $A \with B \leadsto C$ is not equivalent to $A \leadsto (B \leadsto C)$, cf \cite{pym2013semantics}.}
Given a set of propositional atoms \textbf{P}, the language of multiplicative additive linear logic \MALL is defined as
follows.\footnote{We focus on the multiplicative-additive fragment of \LL. Another important part of LL is given by the \textit{exponentials}, 
that allow for retrieving the usual classical inferences in a controlled way. Therefore, instead of being yet another non-classical logic, linear logic  is motivated at least as an analysis of proofs in classical and intuitionistic logic. We leave a discussion of the exponential for JA for future work.}

\begin{equation}
\mathcal{L}_{\MALL} \; ::=  p \in \textbf{P} \mid  \textbf{1} \mid \bot \mid \top \mid \textbf{0} \mid \neg L \mid L \otimes L\mid L \parr L\mid L \oplus L \mid L \with L \mid A \multimap A \mid A \leadsto A
\end{equation}

By dropping weakening and contraction, the units of the logic have to be distinguished in multiplicatives $\textbf{1}$ and $\bot$ (which are units for $\otimes$ and $\parr$ respectivley) and additives ($\top$ and $\textbf{0}$, which are units for $\with$ and $\oplus$ respectively).\footnote{Without W and C, also negation behaves differently. For example, the \textit{ex falso quodlibet} principle is no longer globally valid in linear logic, for the multiplicative false constant $\bot$. Thus, linear logic, beside being non-monotonic, is also a \emph{paraconsistent} logic, at least in the weak sense of invalidating \emph{ex falso quodlibet}. For the sake of simplicity of presentation, we shall use a single notation for negations in various logics.} The sequent calculus for \MALL is presented in Table \ref{sequent:mall}.


\begin{table}[h]
\begin{center}
\begin{center}

\textit{Identities}

\begin{tabular}[]{cc}
& \\
\AxiomC{} \RightLabel{\footnotesize ax} \UnaryInfC{$A \vdash A$} \DisplayProof  
& \AxiomC{$\Gamma, A \vdash \Delta$} \AxiomC{$\Gamma' \vdash A, \Delta'$} \RightLabel{\footnotesize cut}
\BinaryInfC{$\Gamma,
\Gamma' \vdash \Delta, \Delta'$} \DisplayProof
\\
 & \\
\end{tabular}

\textit{Negation}

\begin{tabular}[]{cc}
& \\
\AC{\Gamma \t A, \Delta} \rl{$L\neg$} \UC{\Gamma, \neg A \t \Delta}
\dip   &
\AC{\Gamma, A\t \Delta} \rl{$R\neg$} \UC{\Gamma \t \neg A, \Delta}
\dip
\\
 & \\
\end{tabular}

\textit{Multiplicatives}

\begin{tabular}{ll}
&  \\
\AC{\Gamma \vdash A, \Delta} \AC{\Gamma' \vdash B, \Delta'}
 \LeftLabel{$\otimes$R} \BC{\Gamma, \Gamma' \vdash A\otimes B, \Delta, \Delta'} \dip  & 
 \AC{\Gamma, A, B \vdash \Delta} \rl{$\otimes$L} \UC{\Gamma, A \otimes B \vdash \Delta}
\dip      \\
& \\
\AC{\Gamma, A \vdash \Delta} \AC{\Gamma', B \vdash \Delta'} \LeftLabel{$\parr$L}
\BC{\Gamma, \Gamma', A \parr B \vdash \Delta, \Delta'} \dip  &

\AC{\Gamma \vdash A, B, \Delta} \rl{$\parr$R}
\UC{\Gamma \vdash A \parr B, \Delta}   \dip \\
& \\
\AxiomC{$\Gamma, A \t \Delta$} \AxiomC{$\Gamma', B \t \Delta'$} \LeftLabel{$\multimap$L}
\BinaryInfC{$\Gamma, \Gamma' A \multimap B \t \Delta, \Delta'$} \DisplayProof & 
 \AxiomC{$\Gamma \t B, \Delta$}
\RightLabel{$\multimap$R}
\UnaryInfC{$\Gamma \t A \multimap B,\Delta$} \DisplayProof\\
& \\

\end{tabular}

\textit{Multiplicative units}

\begin{tabular}{cc}
&  \\
\AC{\Gamma \vdash \Delta} \rl{\textbf{1}L} \UC{\Gamma, \textbf{1} \t \Delta}\dip & \AC{} \rl{\textbf{1} R} \UC{\t \textbf{1}}\dip\\
& \\
\AC{} \rl{$\bot$ L} \UC{\bot \t}\dip & \AC{\Gamma \vdash \Delta} \rl{$\bot$ R} \UC{\Gamma \t \Delta, \bot}\dip \\
& \\
\end{tabular}

\textit{Additives}

\begin{tabular}[]{cc}
&  \\
\AxiomC{$\Gamma \t A, \Delta$} \AxiomC{$\Gamma \t B, \Delta $} \LeftLabel{$\with$R}
\BinaryInfC{$\Gamma \t A \& B, \Delta$} \DisplayProof & \AxiomC{$\Gamma, A_{i} \t \Delta$} \RightLabel{$\with$L} \UnaryInfC{$\Gamma, A_{0} \&
A_{1} \t \Delta$} \DisplayProof\\
& \\

\AxiomC{$\Gamma, A \t \Delta$} \AxiomC{$\Gamma, B \t \Delta$} \LeftLabel{$\oplus$L}
\BinaryInfC{$\Gamma, A\oplus B \t \Delta$} \DisplayProof &  \AxiomC{$\Gamma \t A_{i}, \Delta$}
\RightLabel{$\oplus$R}
\UnaryInfC{$\Gamma \t A_{0} \oplus A_{1},\Delta$} \DisplayProof\\
& \\
\AxiomC{$\Gamma \t A, \Delta$} \AxiomC{$\Gamma, B \t \Delta$} \LeftLabel{$\leadsto$L}
\BinaryInfC{$\Gamma, A \leadsto B \t \Delta$} \DisplayProof & 
 \AxiomC{$\Gamma \t B, \Delta$}
\RightLabel{$\leadsto$R}
\UnaryInfC{$\Gamma \t A \leadsto B,\Delta$} \DisplayProof\\
& \\
\end{tabular}

\textit{Additive units}

\begin{tabular}{cc}
&  \\
no rule for ($\top$ L ) & \AC{} \rl{$\top$ R} \UC{\Gamma \t \Delta, \top}\dip\\
& \\
\AC{} \rl{$\textbf{0}$ L} \UC{\Gamma, \textbf{0} \t \Delta}\dip & no rule for ($\textbf{0}$ R ) \\
& \\
\end{tabular}

\end{center}
\caption{Sequent calculus for \MALL}
\label{sequent:mall}
\end{center}
\end{table}
\noindent

We label by \MLL the multiplicative fragment and by \ALL the additive fragment of \MALL.
The intuitionistic version of \MALL, label it by \IMALL, can again be defined by forcing the sequents to contain at most one formula on the right. Moreover,
 \MALL plus weakening is also known as \emph{affine logic}, see \cite{kopylov1995decidability}. 

\subsection{Distributive Substructural logics}
\label{sec:hilbert}

Since we used the notations decided by \cite{Girard1987} to introduce linear logic, we keep this notation also for the other substructural logics, although the scholars in that area usually deploys different notation. A comparison with the notation used in substructural logics is presented in Table \ref{notation}. 

\begin{table}[h]
\begin{center}
\begin{tabular}{c|c}
\cite{Girard1987} & \cite{Restall2002} \\
\hline
$\otimes$  & $\circ$\\
$\parr$ & $+$\\
$\with$ & $\wedge$\\
$\oplus$ & $\vee$\\
$\multimap$ & $\rightarrow$\\
$\textbf{1}$ & 1\\
$\bot$ & 0 \\
$\top$ &  $\top$\\
$\textbf{0}$ &  $\bot$\\
\end{tabular}
\end{center}
\caption{Notations}
\label{notation}
\end{table}

One of the peculiarity of linear logic with respect to other substructural logics is that in linear logic additive connectives are not \emph{distributive}, that is the formulas $A \with (B \oplus C)$ and $(A \with B) \oplus (A \with C)$ are not equivalent in  $\MALL$. By contrast, relevant logics are in general distributive. This makes a critical difference form the point of view of semantics, cf. \cite{Troelstra1992}, and it has also serious consequences on the sequent calculus presentation.
One may be tempted to define a rule of the sequent calculus that entails distributivity of the additives. Unfortunately, this is in general not possible without also assuming weakening and contraction. This is one of the important limitation of the sequent calculus, cf \cite{ciabattoni2012algebraic}. There is a number of ways to extend sequent calculi to cope with that. For instance, one may introduce \emph{hypersequents} or \emph{display calculi} \cite{Paoli2002}. For the present purposes, a presentation of the extensions of $\MALL$ via Hilbert system suffices.\par 
We start by presenting the Hilbert system for $\MALL$, label it by $\HMALL$ (cf. \cite{Troelstra1992}). Then, we introduce its extensions. The concept of deduction of $\HMALL$ requires a tree-structure in order to handle the hypothesis in the correct resource-sensitive way. 

\begin{table}[h]
\begin{enumerate}
\item $\t A \multimap A$
\item $\t (A \multimap B) \multimap ((B \multimap C) \multimap (A \multimap C))$
\item $\t (A \multimap (B \multimap C)) \multimap (B \multimap (A \multimap C))$
\item $\t \neg \neg A \multimap A$
\item $\t (A \multimap B) \multimap (\neg B \multimap \neg A)$
\item $\t A \multimap (B \multimap A \otimes B)$
\item $\t (A \multimap (B \multimap C)) \multimap (A \otimes B \multimap C)$
\item $\t \textbf{1}$
\item $\t \textbf{1} \multimap (A \multimap A)$
\item $\t \textbf{0} \multimap \neg \textbf{1}$
\item $\t \neg \textbf{1} \multimap \neg \textbf{0}$
\item $\t (A \with B) \multimap A$
\item $\t (A \with B) \multimap B$
\item $\t ((A \multimap B) \with (A \multimap C)) \multimap (A \multimap B \with C)$
\item $\t A \multimap A \oplus B$
\item $\t B \multimap A \oplus B$
\item $\t A \wedge B \multimap \neg ( \neg A \oplus \neg B)$, $\neg ( \neg A \oplus \neg B) \multimap A \wedge B$
\item $\t A \parr B \multimap \neg (A \multimap B)$, $ \neg (A \multimap B) \multimap A \parr B$
\item $\t A \otimes B \multimap \neg (\neg A \parr \neg B)$, $\neg (\neg A \parr \neg B) \multimap A \otimes B$
\item $\t (A \multimap C) \with (B \multimap C) \multimap (A \oplus B \multimap C)$
\end{enumerate}
\caption{Axioms of \HMALL}
\label{himall}
\end{table}

The notion of derivation in the Hilbert system for \HMALL is the following.

\begin{definition}[Deduction in \HMALL]\label{def:derivation-mall}
A \emph{deduction tree} in \HMALL  $\mathcal{D}$ is inductively constructed as follows. (i)~The leaves of the tree are assumptions $A \t A $, for $A \in \mathcal{L}_{\MALL}$, or
$\t B$ where $B$ is an axiom in Table \ref{himall} (base cases).\\
 (ii) We denote by $\stackrel{\mathcal{D}}{\Gamma \t A}$ a deduction tree with conclusion $\Gamma \t  A$. If $\mathcal{D}$ and $\mathcal{D}'$ are deduction trees, then the following are deduction trees (inductive steps). 

\begin{center}
\begin{tabular}{ccc}
$
\AC{\stackrel{\mathcal{D}}{\Gamma \t  A}}
\AC{\stackrel{\mathcal{D}'}{\Gamma' \t  A \multimap B}} \rl{$\multimap$-rule}
 \BC{\Gamma, \Gamma' \t  B}
 \dip
$

&

\hspace{10pt}
& 

$
\AC{\stackrel{\mathcal{D}}{\Gamma \t A}}
\AC{\stackrel{\mathcal{D}'}{\Gamma \t B}} \rl{$\with$-rule}
\BC{\Gamma \t A \with B}
\dip
$
\\
\end{tabular}
\end{center}

\end{definition}

Note that Hilbert system \HMALL and the sequent calculus for $\MALL$ are equivalent as expected: $\Gamma \t \Delta$ is provable in the sequent calculus iff the (multiplicative) disjunction of the formulas in $\Delta$ is derivable from $\Gamma$ in \HMALL.\par
Moreover, for \HMALL the deduction theorem holds, that is if $\Gamma, A \t B$, then $\Gamma \t A \multimap B$ (cf. \cite{Troelstra1992}). 
Note that we can also present $\MALL$ plus weakening or contraction by adding to \HMALL the suitable axioms.
From linear logic, relevant logic $\R$ is obtained by adding contraction (C) and distributivity of the additives (D1) and (D2).
Therefore, the Hilbert system for $\R$ is given by axioms 1- 20 plus the following (C), (D1) and (D2), cf. \cite{Paoli2002,Restall2002}. The definition of derivation in $\R$ simply extends our previous Definition \ref{def:derivation-mall}.

\begin{table}[h]
\begin{enumerate}
\item[] (W): $\t A \multimap (B \multimap A)$
\item[] (C): $\t (A \multimap (A \multimap B)) \multimap (A \multimap B)$
\item[] (D1): $\t A \with (B \oplus C) \multimap (A \with B) \oplus (A \with C)$, 
\item[] (D2): $\t (A \oplus B) \with (A \oplus C) \multimap A \with (B \oplus C)$
\end{enumerate}
\end{table}

%
%
%
%

We label the additive fragement of \R by \AR, which includes the additive connectives of \R.
We summarise the relationships between the logics that we have discussed in Figure 1.  The Lambek calculus is the weakest logic that we discuss in this paper. By adding E, we obtain linear logic: by assuming exchange E, the conjunction becomes commutative and the two order-sensitive implications become equivalent. By adding contraction C and distributivity, we obtain relevant logic and by adding also weakening W, we obtain classical logic. By W and C, additives and multiplicatives become equivalent. 


\begin{figure}[h]
\begin{center}
\includegraphics[scale=0.55]{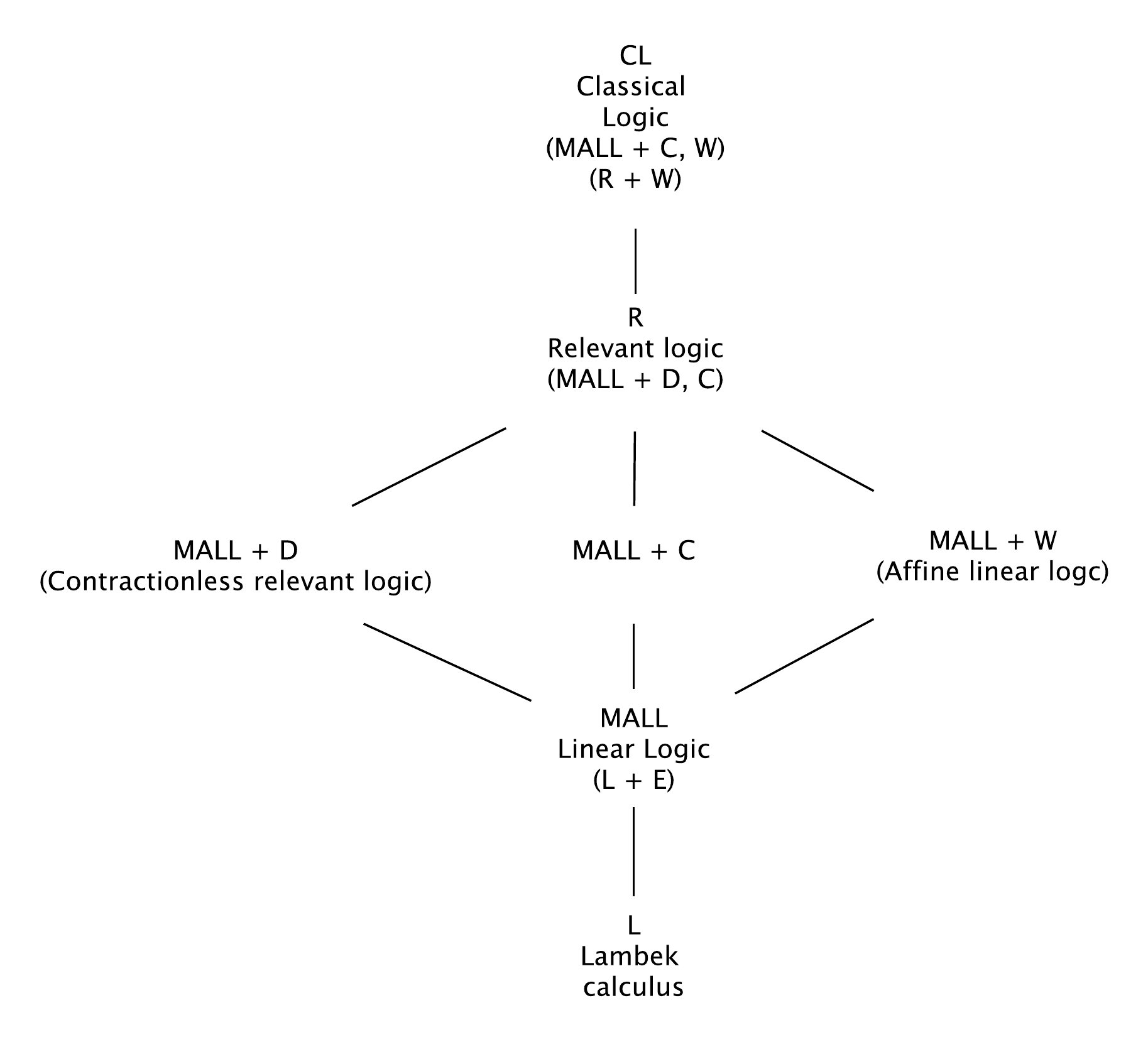}
\end{center}
\caption{Classical and Substructural logics}
\label{figure}
\end{figure}

\section{A model of judgment aggregation for general logics}\label{sec:modelJA}

In this section, we adapt the model of Judgment Aggregation (see \cite{ListPuppe2009,EndrissEtAlJAIR2012}) to cope with the logics that we have introduced. 
We have seen that reasoning about a number of propositions expressed in classical logics amounts to view the them as collected into a set. When discussing logics that lack weakening, contraction, or exchange, we have to replace the notion of a \emph{set} of propositions (or judgments) of the standard JA with the notions of multiset or list of judgments.
Recall that a \emph{multiset} is a pair $(M,f)$ where $M$ is a set and $f: M \to \mathbb{N}$ is a function that assigns to each element of $M$ its multiplicity (i.e. a natural number). A \emph{list} is multiset endowed with a strict total order $\prec$: $(M,f,\prec)$. $(M,f)$ is a \emph{submultiset} of $(M',f')$ iff $M \subseteq M'$ and for every $x \in M$, $f(x) \leq f'(x)$. $(M,f,\prec)$ is a \emph{sublist} of $(M',f',\prec')$ iff it is a submultiset of $(M',f',\prec')$ and $\prec$ is the restriction of $\prec'$ to $M$. 
Moreover, we denote by $\textsf{P}(M)$ the powerset of a multiset (list) $M$, that is, the set of all multisets (lists) that are included in the multiset (list) $M$.\footnote{Possible applications of viewing individual and collective judgments as multisets or lists are the following. Lists may encode propositions in an agenda with temporal dependency or, more generally, with a priority. The case of multisets is motivated as follows. As we discussed in Section 2, if we want to keep track of resource-sensitivity by means of logical reasoning, we have to drop contraction (C). That amounts to assuming that judgments form multisets.  For instance, suppose that agents have to express their opinions on possible \emph{deals} to trade goods (cf. \cite{PorelloEndrissKR2010,PorelloEndrissECAI2010}). A deal that trades a single occurrence of a good $a$ with a single occurrence $b$ can be represented in linear logic by means of a formula $a \multimap b$. The implication states that one token of $a$ can be traded for one token of $b$. In case we want to allow agents to express their opinions about  ``more'' deals between $a$s and $b$s, we may add to the agenda a sufficient number of copies of goods $a$ and $b$ and of the ``deal'' $a \multimap b$.}

Let $N$ be a (finite) set of agents. We shall assume throughout this article that the number of individuals is odd and is bigger than 3.
For every formula $\phi$, we define the \emph{complement} $\sim \phi$ of a formula $\phi$ as follows: $\sim \phi = \neg \phi$, if $\phi$ is not negated, $\sim \phi = \psi$, if $\phi$ is negated and $\phi = \neg \psi$.\footnote{Note that, in intuitionistic logics, double negations do not cancel out. Even in case we were to assume that an agenda in intuitionistic logic is closed under (intuitionistic) negation, the condition of closure under complements does not rule out the presence of doubly negated formulas nor it makes  $\neg \neg \phi$ equivalent to $\phi$.}  An \textit{agenda} $\Phi_{L}$ is a (finite) set (multiset, list) of propositions in the language $\mathcal{L}_L$ of a given logic $L$ (among those that we have previously introduced) that is closed under complements.
Moreover, we assume that the agenda does not contain tautologies nor contradictions (in particular, it does not contain the units of the logic $L$), as usual in the JA literature (cf. \cite{ListPuppe2009}).\par

\begin{definition}
A \textit{judgement set (multiset, list)} $J$ is a (finite) set (multiset, list) of elements of $\Phi_{\textbf{L}}$.  
\end{definition}

We slightly rephrase the usual rationality conditions on judgment sets in terms of derivability in a logic $L$.
With respect to a logic $L$, we say that $J$ is \emph{complement-free} if $J$ does not contain both $\phi$ and the complement of $\sim \phi$; we say that $J$ is \textit{consistent} iff $J \nvdash_{L} \emptyset$\footnote{In case the logic has a formula for expressing absurd, $\bot$, this condition is reformulated as $J\nvdash \bot$}; $J$ is  \textit{complete} iff for all $\phi \in \Phi_{L}$, $\phi \in J$ or $\sim \phi \in J$; and $J$ \textit{deductive closed} iff whenever $J \t_{L} \phi$ and $\phi \in \Phi_{L}$, $\phi \in J$.\par

In JA based on classical logic, the notion of complement-freeness captures a weak form of consistency that is preserved by many aggregation procedures, viz. the majority rule.\footnote{Complement-freeness has been introduced in \cite{EndrissEtAlJAIR2012} and corresponds to the property of weak consistency in \cite{DietrichListJTP2007}.} In classical logic, consistency entails complement-freeness. When dealing with JA in general logics, a significant departure from the standard JA framework is that consistency does not entail complement-freeness any longer.  
This is caused by the possible lack of weakening. Let $L$ be a logic where weakening holds, if a set (multiset) $X$ is inconsistent, then any (mutliset) $X'$  such that $X \subseteq X'$ is also inconsistent. This can be easily shown as follows: if $X \t $ and $X\subseteq X'$, then by weakening we infer $X'\t$.
Thus, the lack of weakening permits that a consistent set $X$ may have inconsistent subsets, which may violate complement-freeness, even in the case of consistency of $X$. As we shall see, this fact has significant consequences on the study of the aggregation of judgments. For this reason, we assume also the following condition.

\begin{definition}
We say that a set (multiset, list) $J$ is \emph{robustly consistent} if $J$ is consistent and every proper subset (submutliset, sublist) $J'$ of $J$ is.  
\end{definition}

Robust consistency always entails consistency and it entails complement-freeness also for logics without weakening. In order to meet the standard treatment of JA, we will always assume that individual judgments sets are robustly consistent, while we shall state explicitly in case we assume that the judgments are also complete. The reason for dropping completeness is that in a number of non-classical logics, e.g. intuitionistic logic and in logics with no classical negation, is not appropriate to express the duality between acceptance and rejection by means of negation. The reason is that positive information and negative information have different statuses in constructive systems.\footnote{Similar motivations for relaxing completeness and closure under complements have been introduced for the case of Description Logics in \cite{PorelloEndrissJLC2014}.}\par

Denote $J(\Phi_{L})$ the set of all robustly consistent judgement sets (multisets, lists) on $\Phi_{L}$. In case we assume that individual judgment sets are also complete, we denote by $J^{*}(\Phi_{L})$ the set of all robustly consistent and complete judgement sets (multisets, lists) on $\Phi_{L}$.  A \textit{profile} of judgements sets $\textbf{J}$ is a vector $(J_{1}, \dots, J_{n})$, where $n = |N|$. In the remainder of the paper, we will occasionally use the characteristic function of a judgment set (multiset or list), with values in $\{0, 1\}$ to visualise profiles by means of a table. By slightly abusing the notation, we will do so also for multisets and lists.\par

\subsection{Aggregation procedures}

We intend to model aggregators that take profiles of judgments that are rational according to a given logic $L$ and return judgements which can be evaluated with respect to a (possibly) different logic $L'$. 
In case $L$ and $L'$ are the same, we are in the standard situation in JA. When the languages of $L$ and $L'$ are different, we would need to define a translation function from the language of $L$ into the language of $L'$. We mainly discuss embeddings of the outcome of an aggregation procedure into weaker or stronger logics that share the same language, with the exception of the discussion in Section 9.1. 

\begin{definition}
An \textit{aggregator procedure} is a function $F: J(\Phi_{L})^{n} \to \textsf{P}(\Phi_{L'})$.
\end{definition}

This definition is quite general and allows in principles for associating profiles of any structure (sets, multisets, lists) to any type of structure (set, multiset, or list). 
In practice, as we shall see in the next sections and as we discussed in Section 2, the type of structure is determined by the logic that is used to assess individual and collective judgments. For instance, the standard JA model associates profiles of sets to sets of judgments. The same shall hold for the case of intuitionistic logic. When discussing JA for the Lambek calculus, aggregation procedures shall associate profiles of lists to lists of judgments and, in the case of linear and relevant logics, the aggregation procedures shall associate multisets to multisets. In Section 9.1, we shall discuss the case of aggregators that connect distinct logics, and in particular we focus on aggregators that associate sets of judgments to multisets of judgments.\par
Note that our definition of aggregation procedure allows for aggregators that return sets (mutlisets, lists) of judgments that are inconsistent w.r.t. $L'$.  This is motivated by the fact that we want to assess the outcome of an aggregation procedure w.r.t.  possibly different logics. That is why the codomain of $F'$ is defined by the set of all sets (multisets, lists) $\textsf{P}(\Phi_{L'})$. \par
In the standard JA framework, an aggregation function depends on the choice of an agenda $\Phi_{L}$, that determines the class of judgment sets, besides on the number of individuals. Here, in the general case, an aggregation function depends on two agendas: the agenda w.r.t. which individuals make their judgments $\Phi_{L}$ and the agenda w.r.t. which the rationality properties of the collective set are evaluated $\Phi_{L'}$. We denote by $[\Phi_{L},\Phi_{L'}]$ the class of aggregation functions from judgments defined on $\Phi_{L}$ to judgments defined in $\Phi_{L'}$, i.e. $[\Phi_{L},\Phi'_{L'}] = \{F \mid F: J(\Phi_{L})^{n} \to \textsf{P}(\Phi_{L'})\}$. Moreover, we use the notation $[\Phi_{L}^{*},\Phi_{L'}]$ when we assume that individual judgments are also complete, i.e.  $[\Phi_{L}^{*},\Phi_{L'}] = \{F \mid F: J^{*}(\Phi_{L})^{n} \to \textsf{P}(\Phi_{L'})\}$. Denote by $N_{\phi}$ the set  $\{i \mid \phi \in J_{i} \}$, the majority rule for $[\Phi_{L},\Phi_{L'}]$ is defined as follows.

\begin{definition}
The majority rule $M : J(\Phi_{\textbf{L}})^{n} \to \textsf{P}(\Phi_{\textbf{L'}})$ is defined by $\phi \in M(\textbf{J})$ iff $|N_{\phi}| > n/2$.
\end{definition}

We also define the following class of aggregation functions that is a simple generalization of the majority rule, the class consists of the \emph{uniform quota rules}, discussed by \cite{DietrichListJTP2007}.

\begin{definition} Let $m \in \{0, \dots, |N| +1\}$, the uniform quota rule defined with quota $m$ is the
aggregation procedure $F_m : J(\Phi_{\textbf{L}})^{n} \to \textsf{P}(\Phi_{\textbf{L'}})$ such that $\phi \in F_{m}(\textbf{J})$ iff $|N_{\phi}| \geq m$. 
\end{definition}

The usual properties of aggregation functions can be rephrased as follows. We start by listing those properties that are usually intended to model fairness desiderata of aggregation procedures. 

\begin{itemize}
 \item \textbf{Anonymity} (A): For any profile $\textbf{J}$ and permutation $\sigma: N \to N$, $F(J_{1}, \dots, J_{n}) = F(J_{\sigma(1)}, \dots, J_{\sigma(n)})$.
 \item \textbf{Neutrality} (N): For any $\phi$ and $\psi$ in $\Phi_{L}$ and profile $\textbf{J}$, if for all $i$ $\phi \in J_{i} \Leftrightarrow \psi \in J_{i}$, then 
 $\phi \in F(\textbf{J}) \Leftrightarrow \psi \in F(\textbf{J})$.
 \item \textbf{Independence} (I): For any $\phi \in \Phi_{L}$ and profiles $\textbf{J}, \textbf{J}' \in J(\Phi_{L})^{n}$, if $\phi \in J_{i} \Leftrightarrow \phi \in J_{i}'$, then $\phi \in F(\textbf{J}) \Leftrightarrow \phi \in F(\textbf{J}')$.
 \item \textbf{Monotonicity} (M): For any $\phi \in \Phi_{L}$ and profiles $\textbf{J} = (J_{1}, \dots,J_{i}, \dots, J_{n})$ 
 and $\textbf{J}' = (J_{1}, \dots,J_{i}', \dots, J_{n})$ if $\phi \notin J_{i}$ and $\phi \in J_{i}'$, then $\phi \in F(\textbf{J}) \Rightarrow \phi \in F(\textbf{J'})$
 \item \textbf{Acceptance-rejection neutrality} (arN): For any $\phi, \psi \in \Phi_{L}$ and any profile $\textbf{J} \in J(\Phi_{L})^{n}$, we have that if
$\phi \in J_{i} \Leftrightarrow \psi \notin J_{i}$ for all agents $i \in N$, then  $\phi \in F(\textbf{J}) \Leftrightarrow \psi \notin F(\textbf{J})$.  
\end{itemize}

(A) states that the aggregator does not favour any particular agent, while (N) implies that it does not favour any particular proposition.
(I) means that the outcome of $F$ w.r.t. a proposition $\phi$ in two different profiles only depends on the patterns of acceptance in the two profiles. 
(M) implies that, by increasing the support of a proposition, $F$ does not change its acceptance. Acceptance-rejection neutrality (arN) has been introduced in \cite{DietrichListSCW2008,DietrichList2009} in order to characterise aggregators in case of weak assumptions concerining individual rationality, namely in case individual judgements sets are just assumed to be consistent. (arN) means that the aggregator is not biased either for or against the acceptance of any proposition.\footnote{This version of acceptance-rejection neutrality is due to \cite{EndrissEtAlJAIR2012}.} \par

The notion of collective rationality that is standard in JA (cf. \cite{ListPuppe2009,EndrissEtAlJAIR2012}) states that an aggregation procedure is collectively rational iff for every profile, the output of the aggregation procedure is consistent, complete, and deductively closed w.r.t. classical logic. Here, we make explicit the logic with respect to which the aggregation is assessed, that is, we define the rationality conditions with respect to the logic that is used to evaluate the output of the aggregation. 

\begin{itemize}
\item $F \in [\Phi_{L},\Phi_{L'}]$ is \textbf{complement-free} iff for every $\textbf{J}$, $F(\textbf{J})$ is complement-free w.r.t. $\textbf{L}'$.
\item $F \in [\Phi_{L},\Phi_{L'}]$ is \textbf{consistent} w.r.t. iff for every $\textbf{J}$, $F(\textbf{J})$ is consistent w.r.t. $\textbf{L}'$.
\item $F \in [\Phi_{L},\Phi_{L'}]$ is \textbf{deductively closed} iff for every $\textbf{J}$, $F(\textbf{J})$ is deductively closed w.r.t. $\textbf{L}'$.
\item $F \in [\Phi_{L},\Phi_{L'}]$ is \textbf{complete}  iff for every $\textbf{J}$, $F(\textbf{J})$ is complete w.r.t. $\textbf{L}'$ 
\item $F \in [\Phi_{L},\Phi_{L'}]$ is \textbf{weakly rational} (WR) iff for every $\textbf{J}$, $F(\textbf{J})$ is complete and complement-free w.r.t. $\textbf{L}'$
\item $F \in [\Phi_{L},\Phi_{L'}]$ is \textbf{robustly consistent} iff for every $\textbf{J}$, $F(\textbf{J})$ is robustly consistent wr $\textbf{L}'$
\end{itemize}

In case we do not assume that individual judgments are complete, we confine ourselves to the study of the preservation of (robust) consistency. 


Let $AX$ be a list of axioms among those above. We denote by $[\Phi_{L},\Phi_{L'}](AX)$ the set of aggregation functions defined with domain $J(\Phi_{L})^{n}$ and range $J(\Phi_{L'})$ that satisfy the axioms in $AX$.\par

In \cite{EndrissEtAlJAIR2012}, two classes of aggregation functions can be characterized in terms of the axioms above. The first class just contains the majority rule, hence it is the characterization of the majority rule, which adapts May's theorem for the case of JA, cf. \cite{MayEconometrica1952}. 
The second class is obtained by dropping weak rationality (WR) and corresponds to the class of uniform quota rules.  
The following proof adapts the one in \cite{EndrissEtAlJAIR2012} for the case of the logics that we have introduced. The significant difference is that, in order to show that the majority rule satisfies (WR), we have to assume that the individual judgment sets are robustly consistent.

\begin{theorem}\label{majority} $F$ is the majority rule iff $F \in [\Phi_{\textbf{L}}^{*},\Phi_{\textbf{L}'}](WR,A,I,N,M)$
\end{theorem}

\begin{proof}
From left to right, the majority rule satisfies the axioms. We only show that the majority rule is weakly rational. Assume that $|N| = n$. For complement-freeness, suppose, by contradiction, that there exists a $\phi$ in $\Phi_{L}$ such that both $\phi$ and $\sim \phi$ are in $M(\textbf{J})$. Then 
$|N_{\phi}| \geq \frac{n+1}{2}$ and $|N_{\sim \phi}| \geq \frac{n+1}{2}$. This entails that there exists an individual $i$ such that $\phi$ and $\sim \phi$ are in $J_{i}$, against the assumption that each $J_i$ is robustly consistent.\footnote{Assuming only consistency of the individual sets is not sufficient, for instance $\{A, \neg A, C\}$ is consistent in \MALL but not complement-free.}\par
The majority rule is also complete. Suppose by contradiction that neither $\phi$ nor $\sim \phi$ are in $M(\textbf{J})$. Then $|N_{\phi}| < \frac{n+1}{2}$ and $|N_{\sim \phi}| < \frac{n+1}{2}$, which entails that there exists an $i$ such that $J_i$ violates completeness.\par
From right to left. Assume $F$ satisfies the axioms above. Since $F$ satisfies $(A)$,$(I)$,$(N)$, the outcome of $F$ only depends on the cardinality of the set of individuals accepting $\phi$ (see also \cite{ListPuppe2009}). That is, $F$ can be represented by a function $h : \{0,\dots, n=|N|\} \to \{0,1\}$ such that $\phi \in F(\textbf{J})$ iff $h(|N_{\phi}|) = 1$. Since $F$ satisfies (M), if $h(i) = 1$ and  $j\geq i$, then also $h(j)=1$. Suppose then that $k$ is the minimum for which $h(k) =1$. Since $F$ satisfies (WR), $F$ must be complete, we get that $k \leq \frac{n+1}{2}$, otherwise there are profiles that lead to incomplete judgment sets. Since $F$ has to be complement-free, we get $k \geq \frac{n+1}{n}$, to avoid acceptance of a formula and its negation.
Thus, $k = \frac{n+1}{2}$, hence $F$ is the majority rule.
\end{proof}
 
Note that there is no mention of preserving (robust) consistency at this point, Theorem \ref{majority} only shows that the majority rule preserves complement-freeness and completeness. By assuming mere consistency instead of robust consistency of the individual judgments, as in the standard JA result, cf.  \cite{ListPuppe2009}, the proof fails for the case of logics without weakening. For instance, suppose $\textbf{J}$ is composed of $n$ copies of $J_i = \{A, \neg A, B\}$: each $J_i$ is consistent but not complement-free, hence the majority $M(\textbf{J})$ would violate complement-freeness as well.\par 
The class of uniform quota rules is characterized as follows. Since the rationality conditions do not enter the proof, we can simply adapt the proof in \cite{EndrissEtAlJAIR2012} to the present framework.
 
 \begin{proposition}
 $F$ is a uniform quota rule iff $F \in [\Phi_{\textbf{L}},\Phi_{\textbf{L}'}](A,I,N,M)$
 \end{proposition}

%
%
%
%
%
%
%

We adapt now the concept of safety of an agenda, that is due to \cite{EndrissEtAlJAIR2012}. Since we are assuming that the individual agenda and the collective agenda might differ, the concept of safety applies to a pair of agendas. 

\begin{definition}
For any set of axioms $AX$, a pair of agendas $(\Phi_{L}, \Phi_{L'})$ is \emph{safe} for axioms $AX$ iff every $F$ in $[\Phi_{L},\Phi_{L'}](AX)$ is robustly consistent.
\end{definition}

The safety of an agenda amounts to assuming that every judgment aggregation problem defined on that agenda $\Phi_{L}$ that uses aggregators of the given class preserves (robustly) consistent outcomes in $L'$. In case $L = L' = \CL$ we obtain the standard definition of safety \cite{EndrissEtAlJAIR2012}.\par
Since we are assessing JA in a variety of logics, it is useful to present the following definition of safety of logics for sets of axioms. 

\begin{definition}
A pair of logics $(L,L')$ is \emph{safe} for axioms $AX$ iff every pair of agenda $(\Phi_{L}, \Phi_{L'})$ is \emph{safe} for axioms $AX$ (i.e. for every pair of agendas $(\Phi_{L}, \Phi_{L'})$, every $F$ in $[\Phi_{L},\Phi_{L'}](AX)$ is robustly consistent.)
\end{definition}

The concept of safety of logics amounts to assuming that for every possible agenda defined by means of the language of the logic $L$ and for every profile of judgments, the aggregation function $F$ preserves robustly consistent judgments defined w.r.t logic $L$ when evaluated w.r.t. the logic $L'$.\par
In case the aggregation procedure is defined w.r.t. a single agenda and w.r.t. a single logic, we say that the (single) agenda $\Phi_{L}$ is safe for the class of axioms $AX$ and that the (single) logic $L$ is. 
Note that concept of safety still applies to classes of axioms that define a single procedure, e.g. to the case of the majority rule that is defined by the axioms WR, A, I, N, and M. In that case, a possibility result --- that states the \emph{existence} of a procedure that satisfies a certain number of axioms and preserves consistency --- and safety results --- that state that \emph{every} function in a certain class preserve consistency --- coincide. We label the class of axioms that characterise the majority rule by MAJ. \par

The last concepts of the standard theory of judgment aggregation that we rephrase for this setting is the following list of properties of agendas. Here we simply generalise it to cope with a variety of logics.\par

Recall that a set (multiset, list) $Y$ of formulas of $L$ is inconsistent w.r.t. $L$ iff $Y \t \emptyset$. $Y$ is minimally inconsistent iff $Y$ is inconsistent and every subset (submultiset, sublist) of $Y$ is consistent. 

\begin{definition}\label{agendas} The following properties define classes of agendas:
\begin{itemize}
\item An agenda $\Phi_{L}$ has the \emph{median property} (MP) iff every minimally inconsistent set (multiset, list) of formulas $Y$ of $\Phi_{L}$ has cardinality at most 2. 
\item An agenda $\Phi_{L}$ has the \emph{simplified median property} (SMP) iff every (non-trivially) inconsistent set (multiset, list) of formulas $Y$ of $\Phi_{L}$ has a subset (submultiset, sublist) $\{\phi, \psi\}$ such that $\phi$ and $\neg \psi$ are logically equivalent: $\phi \t_{L} \neg \psi$ and $\neg \psi \t_{L} \phi$.
\item An agenda $\Phi_{L}$ has the \emph{syntactic simplified median property} (SSMP) iff every (non-trivially) inconsistent set (multiset, list) of formulas $Y$ of $\Phi_{L}$ has a subset (submultiset, sublist) $\{\phi, \neg \phi\}$.
\item An agenda  $\Phi_{L}$  has the $k$-median property (kMP) iff every minimally inconsistent set (multiset, list) of formulas $Y$ of $\Phi_{L}$ has cardinality at most $k$.
\end{itemize}
\end{definition}

The median property is the case with $k = 2$ of the $k$-median property. Moreover, the SSMP entails SMP which in turn entails MP. The opposite directions do not hold.  The \emph{median property} is due to \cite{NehringPuppeJET2007}. As we shall see, the median property characterizes, in case of classical logic, the agendas that are safe for the majority rule. The other properties are adequate to characterize agendas for larger classes of aggregators, cf. \cite{EndrissEtAlJAIR2012,ListPuppe2009}.

\subsection{Summary of results}

We summarise in the following table the results that we are going to establish in the subsequent sections concerning the safety of logics and agendas for a set of axioms. The first line of Table 6 simply rephrases the results about classical logic and classical agendas known from JA. 
%

\begin{table}[h]
\footnotesize
\begin{tabular}{|l|c|c|c|c|c|}
\hline
& MAJ (WR, A, I, N, M) & Quota rules (A, I, N, M)  & (WR, A, N ,I)  &  (WR, A , N) & (WR, A, I)\\
\hline
$\Phi_{\CL}$ & safe iff MP & safe iff kMP  & safe iff SMP & safe iff SMP & safe iff SSMP\\
$\Phi_{\IL}$  & safe iff MP & safe iff kMP  & safe iff SMP & safe iff SMP & safe iff SSMP\\
$\Phi_{\L}$  & safe iff MP & safe iff kMP  & safe iff SMP & safe iff SMP & safe iff SSMP\\
$\Phi_{\MALL}$  & safe iff MP & safe iff kMP  & safe iff SMP & safe iff SMP & safe iff SSMP\\
$\Phi_{\MLL}$  & safe iff MP & safe iff kMP  & safe iff SMP & safe iff SMP & safe iff SSMP\\
$\Phi_{\ALL}$  & always safe & safe with $m \geq n / 2$ & safe iff SMP & safe iff SMP & safe iff SSMP\\
$\Phi_{\ALL W}$  & safe iff MP & safe iff kMP  & safe iff SMP & safe iff SMP & safe iff SSMP\\
$\Phi_{\ALL C}$  & always safe & safe with $m \geq n / 2$  & safe iff SMP & safe iff SMP & safe iff SSMP\\
$\Phi_{\R}$ & safe iff MP & safe iff kMP  & safe iff SMP & safe iff SMP & safe iff SSMP\\
$\Phi_{\AR}$  & always safe & safe with $m \geq n / 2$  & safe iff SMP & safe iff SMP & safe iff SSMP\\
      & & & & &\\
\hline
\end{tabular}
\caption{Summary of results concerning the safety of agendas and logics for sets of axioms.}
\label{tab:summary}
\end{table}

As we shall discuss, all the monotonic logics (i.e. the logics whose sequent calculus admits weakening W)  that are listed in the table exhibit, regarding the safety of the agenda, a situation that is analogous to that of classical logic. This is expected, due to the results in \cite{Dietrich2007}.  However, dropping monotonicity is not sufficient to achieve safety or possibility results, as the situation of a number of non-monotonic logic shows (cf. \MALL, \MLL, \R). 
The case of $\ALL$ and $\AR$ are significant here: since every agenda in \ALL or \AR is safe for MAJ, we can claim that the logics \ALL and \AR are indeed safe for those axioms. In particular, as we shall see, the majority rule is robustly consistent for any agenda in \ALL or \AR. Those systems are indeed non-monotonic however, to achieve safety, as we shall see, we have to restrict to the additive fragment of those logics. 

Observe that when an agenda is not safe for a certain class of axioms AX, this means that there exists an aggregation procedure in the class of functions defined by means of AX that is not robustly consistent. Thus, if an agenda is not safe for axioms AX, this entails that the agenda is not safe for any subset of AX. Moreover, if a logic is not safe for a certain class of axioms AX, it means that there exists an agenda $\Phi_{L}$, such that it exists an aggregation function $F \in [\Phi_{L}, \Phi_{L}](AX)$ that is not robustly consistent. Namely, a possible outcome $Y$ of $F$ is inconsistent in $L$. If $Y$ is inconsistent in $L$ (that is $Y \t_{L} \emptyset$), then $Y$ is inconsistent in any logic $L'$ that is stronger (i.e. that proves more sequents) than $L$ (i.e. $Y \t_{L'} \emptyset$). Therefore, if a logic is not safe for AX, then any stronger logic is not safe for AX.

\section{Judgment Aggregation in extensions of classical logic}\label{sec:JAclassical}

For classical logic $\CL$, we assume that every judgment set is also complete. Moreover, robust consistency and consistency in this case coincide. In our setting, List and Pettit's result can be rephrased as follows.

\begin{theorem}[List and Pettit]\label{ListPettit2002}
There are agendas defined in $\CL$ that are not safe for $MAJ$ (i.e. for axioms (WR, A, I, N, M)). 
\end{theorem}

For instance, the agenda $\{A,  B, A \wedge B, \neg A, \neg B, \neg(A \wedge B)\}$ provides the famous discursive dilemma (cf. \cite{ListPettit2002,KornhauserSagerCLR1993}). That is, on that agenda, there is in fact a profile $\textbf{J}$ such that the majority rule returns an inconsistent set. 

 \begin{center}
\begin{tabular}{ccccccc}
 & $A$ & $A \wedge B$ & $B$ & $\neg A$ & $\neg (A \wedge B)$ & $\neg B$  \\
\hline
$i_{1}$ & 1 & 1 & 1 & 0 & 0 & 0 \\   
$i_{2}$ & 1 & 0 & 0 & 0 & 1 & 1\\ 
$i_{3}$ & 0 & 0 & 1 & 1 & 1 & 0\\
\hline 
maj. & \textbf{1} & 0 & \textbf{1} & 0 & \textbf{1} & 0\\
\end{tabular}
\end{center}

The collective set $F(\textbf{J}) = \{ A, B, \neg(A \wedge B)\}$ is not consistent in classical logic. In proof-theoretic terms, it means that $ \{ A, B, \neg(A \wedge B)\} \t $, and that can be shown as follows.

\begin{center}
\AC{A \t A}
\AC{B \t B}
\rl{$\wedge$L}
\BC{A,B \t A \wedge B}
\rl{$\neg$L}
\UC{A, B, \neg(A \wedge B) \t \empty}
\DisplayProof
\end{center}

Theorem \ref{ListPettit2002} can be extended to various classes of functions, in particular all those classes that include the majority rule (cf. \cite{ListPuppe2009,EndrissEtAlJAIR2012}). 

Thus, by playing with our definitions, we can infer that classical logic is not safe for the majority rule and for any class of axioms that define procedures that include the majority rule.


Because of Theorem \ref{ListPettit2002}, in JA it is important to characterize which type of agendas are safe for a certain set of axioms. 
In particular, Theorem \ref{ListPettit2002} can be refined by saying that the agendas that are not safe are those that violates the median property. 

\begin{theorem}[\cite{NehringPuppeJET2007}]
An agenda $\Phi_{CL}$ satisfies the median property iff $\Phi_{CL}$ is safe for MAJ.
\end{theorem}

For larger classes of functions, the median property is no longer sufficient. The following proposition summarizes the relationships between agenda properties and classes of functions for classical logic. 

\begin{proposition}\label{cl:SSMP}[\cite{EndrissEtAlJAIR2012}]\ The following facts hold:
\begin{itemize}
\item $\Phi_{\CL}$ satisfies the SMP iff it is safe for $(WR, A, N, I)$
\item $\Phi_{\CL}$ satisfies the SMP iff it is safe for $(WR, A, N)$
\item $\Phi_{\CL}$ satisfies the SSMP iff it is safe for $(WR, A, I)$
\item $\Phi_{\CL}$ satisfies the kMP iff it is safe for $(A, I, N,M)$,  (for $F_{m}$, with $m > n - \frac{n}{k})$
\end{itemize}
\end{proposition}

For the class of uniform quota rules, note that the choice of the threshold is crucial as well to preserve consistency. 
The median property is the condition that guarantees that an agenda in classical logic is safe for the majority rule. Since the median property is defined in terms of minimally inconsistent sets of formulas, any logic that conservatively extends classical propositional logic shall suffer the same problems of aggregation. If $X$ is minimally inconsistent in classical logic, then $X$ is minimally inconsistent in any conservative extension of classical logic. 
Therefore, there is no hope to mend propositional inconsistency by enriching the language of the logic.

\begin{corollary}
Any extension of classical logic is not safe for MAJ.
\end{corollary}

For instance, any modal logic and any description logic, as soon as they extend propositional reasoning, are not safe (cf. for instance \cite{PorelloEndrissJLC2014} for description logics, and \cite{Pauly2007}, \cite{EndrissGrandiECAI2014} for general Kripke structures).


\section{Judgment aggregation in Intuitionistic Logic}\label{sec:il}

We start the study of judgment aggregation in non-classical logics by dealing with intuitionistic logic.  For \IL, we do not assume that individual judgments sets are complete. Recall that intuitionistic negation is defined by $\neg A = A \rightarrow \bot$.\footnote{Assuming that the agenda may contain $A \rightarrow \bot$ does not entail that the agenda contains the formula for false $\bot$. That is, in this case, the agenda is not closed under the atoms occurring in the formulas. Such an agenda is called \emph{non truth-functional} in \cite{nehring2008consistent}.} For intuitionistic logic, weakening holds, hence consistency and robust consistency are equivalent. 
We can easily see that with respect to judgment aggregation, intuitionistic logic does not significantly differ from classical logic. 

\begin{theorem}
Intuitionistic logic $\IL$ is not safe for MAJ. 
\end{theorem}

\begin{proof} It is sufficient to exhibit an agenda w.r.t. which an aggregation problem generates an inconsistent set. 
For instance, we show that any agenda that includes $\{A,  B, A \wedge B, (A \wedge B) \rightarrow \bot\}$ is not safe for the majority rule in intuitionistic logic. There exists indeed a profile of judgments $\textbf{J}$ (adapt the one we encountered in Section \ref{sec:JAclassical}) such that $F(\textbf{J}) = \{A, B, (A \wedge B) \rightarrow \bot\}$. This set is inconsistent in intuitionistic logic, as the following proof shows.

\begin{center}
\AC{A \t A}
\AC{B \t B}
\rl{$\wedge$R}
\BC{A,B \t A \wedge B}
\AC{\bot \t \bot}
\rl{$\rightarrow$L}
\BC{A, B, (A \wedge B) \rightarrow \bot \t \bot}
\DisplayProof
\end{center}

\end{proof}

Also in the case of intuitionistic logic, the median property is the appropriate condition that characterizes safe agendas for the majority rule. The following proof largely adapts the known result for the case of non-complete judgments sets (cf. \cite{PorelloEndrissJLC2014}). 

\begin{theorem}\label{th:ilmp}
An agenda $\Phi_{\IL}$ is safe for MAJ iff it satisfies the median property.
\end{theorem}

\begin{proof}  In one direction, we show that if $\Phi_{\IL}$ satisfies the median property, then the majority rule is consistent.  Suppose by contradiction that there is an agenda $\Phi_{\IL}$ that satisfies the median property and that there is a profile $\textbf{J}$ such that $F(\textbf{J})$ is inconsistent. Since the median property holds, if $F(\textbf{J})$ is inconsistent, then there is a minimally inconsistent set $Y \subseteq F(\textbf{J})$ with cardinality at most 2. Firstly, $Y$ cannot have cardinality 1, otherwise there must be a contradictory formula in the agenda.\footnote{In case we do assume contradictions in the agenda, we can reason as follows: if $|Y|=1$ there should be a majority of agents for which $Y$ is in $J_i$, violating the assumption of consistency of the individual sets $J_i$}
Suppose then that $Y = \{\phi_1, \phi_2\}$. Since $Y$ is accepted by majority, we have that $|\{i \in N \mid \phi_1 \in J_i\}| \geq \frac{n+1}{2}$ and $|\{i \in N \mid \phi_2 \in J_i\}| \geq \frac{n + 1}{2}$. This entails that there is an individual $j$, such that  both $\phi_1$ and $\phi_2$ are in $J_j$, against assumption that every individual judgment set is consistent. \par
In the other direction, we prove the contrapositive statement: if $\Phi_{\IL}$ violates the median property, then $F(\textbf{J})$ is inconsistent. Suppose that $\Phi_{\IL}$ violates the median property, then there exists a subset $Y$ that is minimally inconsistent of size strictly bigger than 3. We construct a profile that violates the consistency of the majority rule.
Suppose $|Y| \geq 3$ and that $\phi$ and $\psi$ are distinct formulas in $Y$. The first $\frac{n-1}{2}$ individuals accept all formulas of $Y$ but $\phi$, the last $\frac{n+1}{2}$ individuals accept all formulas of $Y$ but $\psi$, and the individual $\frac{n+1}{2}$ accepts just $\phi$ and $\psi$. Each individual set is consistent, however, by majority, $Y$ is accepted. That is, $Y$ it is contained in $M(\textbf{J})$. Since $\IL$ satisfies $W$, if $Y$ is inconsistent, then $M(\textbf{J})$ is inconsistent.  
\end{proof}

Note that we can use the sole consistency assumption to conclude the argument that shows that the median property is necessary because $\IL$ satisfies weakening.\par 
By enlarging the class of aggregation procedures beyond the majority rule, that is by focusing on uniform quota rules, we can show that the situation is similar w.r.t. the classical case.

\begin{proposition}  $\Phi_{\IL}$ satisfies the kMP iff it is safe for (A, I, N,M), for $F_{m}$ with $m > n - \frac{n}{k}$.
\end{proposition}

The proof largely adapt the case for classical logic (cf \cite{DietrichListJTP2007}, Corollary 2), so we omit it.

\section{Judgment Aggregation in the Lambek Calculus}\label{sec:lambek}

We start discussing substructural logics by studying the Lambek calculus $\L$.  In this case, we do not assume that individual judgments are complete.\footnote{In fact, the Lambek calculus is an intuitionistic logic, \cite{Abrusci1990}.} Individual and collective judgments are (finite) lists of formulas $J = [ \phi_1, \dots, \phi_m]$ that are sublists of $\Phi_{\L}$. Since the aggregation procedures for \L return lists of judgments, we have to be explicit in defining how the aggregation computes the outcome. We present the case of the majority rule, other aggregation procedures can be handled in a similar way. We have to keep track of the positions of the formulas in the $J_{i}$s with respect to the order of formulas in $\Phi_{\L}$. Denote by $\pi_{j}(X)$ the $j$-th element of the list $X$. We write a judgment set $J$ by means of $J'= J \cup \{\star\}$, where $\star$ is a designated symbol. $J'$ is defined by $\pi_{j}(J')= \phi$ if $\phi \in J$ and $j$ is the position of $\phi$ in $\Phi_{\L}$, otherwise $\pi_{j}(J) = \star$. For example, the sublist $[B,D]$ of $[A,B,C,D]$ can be written by $[\star, B, \star, D]$.\par
 Given a profile $\textbf{J}$, let $N^{j}_{\phi} = \{i \mid \phi = \pi_j (J_{i})\}$, that is, $N^{j}_{\phi}$ denotes the set of individuals that place $\phi$ at the $j$-th position of their list of judgements.\par
Define the list $[x_{1} ,\dots, x_{l}]$, where $l$ is the length of the list $\Phi_{\L}$ and each $x_j$ is either a formula of the agenda of the designated symbol $\star$: $x_j = \phi_{j}$ if $|N^{j}_{\phi_{j}}| \geq \frac{n + 1}{2}$ and $x_{j} = \star$, if $|N^{j}_{\phi_{j}}| <  \frac{n + 1}{2}$. $M(\textbf{J})$ is then the sublist of $[x_{1} ,\dots, x_{l}]$ that is obtained by removing all the $\star$ symbols.\par
It is easy to see that collective inconsistencies may emerge also for \L.

\begin{theorem}
The Lambek calculus $\L$ is not safe for MAJ. 
\end{theorem}

\begin{proof}

Take any agenda that includes $[A, B, A \odot B,(A \odot B) \setminus \bot]$ and a profile of judgments $\textbf{J}$ as follows.

\begin{center}
\begin{tabular}{ccccc}
 & $A$ & $B$ & $A \odot B$ & $A \odot B \setminus \bot$ \\
\hline
$i_{1}$ & 1 & 1 & 1 & 0\\   
$i_{2}$ & 1 & 0 & 0 & 1\\ 
$i_{3}$ & 0 & 1 & 0 & 1\\
\hline 
maj. & \textbf{1} & \textbf{1} & 0 & 1\\
\end{tabular}
\end{center}

The outcome of the majority rule on $\textbf{J}$ is the list $[A, B, A \odot B \setminus \bot]$, which is not consistent in $\L$, as the following proof shows.

\begin{center}
\AC{A \t A}
\AC{B \t B}
\rl{$\odot$L}
\BC{A,B \t A \odot B}
\AC{\bot \t \bot}
\rl{$\setminus$ L}
\BC{A, B, A \odot B \setminus \bot \t \bot}
\DisplayProof
\end{center}
\end{proof}

In $\L$ consistency and robust consistency are \emph{not} equivalent. For instance, $[A, A\setminus \bot, C]$ is consistent, since it does not entail $\bot$, whereas $[A, A \setminus \bot]$ is inconsistent. The following example shows that the majority rule for $\L$ may return an inconsistent list even in an extremely simple agenda.

\begin{example}
Take the agenda $\{A, A \setminus \bot, C,  C \setminus \bot\}$ and the following profile of (incomplete) lists of judgments. 

\begin{center}
\begin{tabular}{cccc}
 & $A$ & $C$ & $A\setminus \bot$ \\
 \hline
$i_1$ & 1 & 1 &  1 \\
$i_2$ & 1 & 0 &  0 \\
$i_3$ & 0 & 0 &  1\\
\hline
maj. & 1 &  0 & 1 
\end{tabular} 
\end{center}

Each individual judgment list is consistent, however the majority returns the inconsistent list $[A, A \setminus \bot]$.

\end{example}

Therefore, in case we assume just the consistency of the individual judgments, the median property is not sufficient to guarantee the consistency of every collective outcomes. The condition of robust consistency is required, as the following proof shows. 

\begin{theorem}\label{th:mpl}
An agenda $\Phi_{\L}$ satisfies the median property iff $\Phi_{\L}$ is safe for MAJ. 
\end{theorem}

\begin{proof}
From left to right, suppose by contradiction that $\Phi_{\L}$ satisfies the median property and $M(\textbf{J})$ is inconsistent. Then, $M(\textbf{J})$ contains a minimally inconsistent list $Y$ of size smaller or equal than 2. $|Y|$ cannot be 1, since we excluded contradictions in the agenda. Thus $Y = [\phi_1, \phi_2]$, and suppose that $\phi_{1}$ and $\phi_2$ take the $j$-th and $j + 1$-th positions (respectively) in the list $M(\textbf{J})$. Therefore, there exist $h$ and $k$, with $h < k$, such that $|N^{h}_{\phi_1}| \geq \frac{n+1}{2}$ and  $|N^{k}_{\phi_1}| \geq \frac{n+1}{2}$, which entails that there is an individual $i$ such that $[\phi_1, \phi_2]$ is a sublist of $J_{i}$.
Note that  $\phi_1$ and $\phi_2$ may be contiguous in $J_i$ or not, however, in both cases, this contradicts the assumption of robust consistency.\par 
From right to left, we show that if $\Phi_{\L}$ violates the median property, then the majority rule is inconsistent. Assume that $\Phi_{\L}$ contains a list $Y$ that is minimally inconsistent of size bigger than 3. Define a profile $\textbf{J}$ that coincides with the one in the proof of Theorem \ref{th:ilmp} for the formulas in $Y$ and, for the other formulas of the agenda, the individuals reject all of them. Note the judgment sets are now incomplete. By construction, we get $M(\textbf{J}) = Y$, which is inconsistent in $\L$. 
\end{proof}

For the larger class of uniform quota rules, we can adapt the known results, by noticing that in this case, again, the safety result applies to the condition of robust consistency.

\begin{proposition} $\Phi_{\L}$ satisfies the kMP iff it is safe for $(A, I, N,M)$, for $F_m$ with $m > n - \frac{n}{k}$.
\end{proposition}

\section{Judgment Aggregation in Linear Logic}\label{section:mall}

We show that for linear logic analogous impossibility results can be stated. We assume that the judgments form multisets of formulas and that they are complete. In \MALL, due to the lack of weakening, consistency and robust consistency are not equivalent. 

\begin{theorem}\label{MALL}
(Multiplicative additive) linear logic $\MALL$ is not safe for MAJ.  
\end{theorem}

\begin{proof}
Take the agenda $\Phi_{\MALL} = \{A, \neg A, B, \neg B, A \otimes B, \neg (A \otimes B)\}$ and the following profile:

 \begin{center}
\begin{tabular}{ccccccc}
 & $A$ & $A \otimes B$ & $B$ & $\neg A$ & $\neg (A \otimes B)$ & $\neg B$  \\
\hline
$i_{1}$ & 1 & 1 & 1 & 0 & 0 & 0 \\   
$i_{2}$ & 1 & 0 & 0 & 0 & 1 & 1\\ 
$i_{3}$ & 0 & 0 & 1 & 1 & 1 & 0\\
\hline 
maj. & \textbf{1} & 0 & \textbf{1} & 0 & \textbf{1} & 0\\
\end{tabular}
\end{center}

The outcome of the majority rule on $\textbf{J}$ is the multiset $\{A, B, \neg (A\otimes B)\}$ which is inconsistent w.r.t. $\MALL$:

\begin{center}
$$
\AC{A \t A}
\AC{B \t B}
\rl{$\otimes$ R}
\BC{A, B \t A \otimes B}
\rl{$\neg$ L}
\UC{A, B, \neg (A \otimes B) \t}
\dip
$$

\end{center}

\end{proof}

Again, the median property is not sufficient to preserve consistency, in case we assume that individual judgments are just consistent. 

\begin{example}\label{ex:mpmall}
Take the agenda closed under complements $\{A, \neg A, C, \neg C, C \multimap \textbf{1}, \neg (C \multimap \textbf{1})\}$. The agenda satisfies the median property, since there is no inconsistent multiset of size strictly bigger than 2. Take the following profile.

\begin{center}
\begin{tabular}{ccccccc}
 & $A$ & $\neg A$ & $C$ & $C \multimap \textbf{1}$ & $\neg (C \multimap \textbf{1})$ & $\neg C$  \\
\hline
$i_{1}$ & 1 & 1 & 1 & 0 & 0 & 0 \\   
$i_{2}$ & 1 & 0 & 1 & 1 & 0 & 0\\ 
$i_{3}$ & 0 & 1 & 1 & 1 & 0 & 0\\
\hline 
maj. & \textbf{1} & \textbf{1} & \textbf{1} & \textbf{1} & 0 & 0\\
\end{tabular}
\end{center}

Each individual multiset is consistent, although the first one is not robustly consistent, as it contains $\{A, \neg A\}$. By majority, we obtain the multiset $\{A, \neg A, C, C \multimap \textbf{1}\}$, which is inconsistent in \MALL, as the following proof shows.

$$
\AC{A\t A}
\rl{$\neg$ L}
\UC{A, \neg A \t}
\rl{\textbf{1} L}
\UC{A, \neg A, \textbf{1} \t}
\AC{C \t C}
\AC{\textbf{1} \t \textbf{1}}
\rl{$\multimap$ L}
\BC{C, C \multimap \textbf{1} \t \textbf{1}}
\rl{cut}
\BC{A, \neg A, C, C \multimap \textbf{1} \t}
\dip
$$

Therefore, the median property is not sufficient to guarantee the consistency of the majority rule in case of \MALL.
\end{example}

The condition of robust consistency is required also for $\MALL$.

\begin{proposition} 
$\Phi_{\MALL}$ satisfies the median property iff it is safe for MAJ.
\end{proposition}

\begin{proof}
From left to right, suppose $\Phi_{\MALL}$ satisfies the median property and $M(\textbf{J})$ is inconsistent. The argument is then similar to the proof of Theorem \ref{th:mpl}.\par
From right to left, assume that $\Phi_{\MALL}$ violates the median property, we show that $M(\textbf{J})$ is not robustly consistent. If $\Phi_{\MALL}$ violates the median property, then there exists an inconsistent set $Y$ included in $\Phi_{\MALL}$ with cardinality bigger than 3. 
Define a profile as in the proof of Theorem \ref{th:ilmp}. We can conclude that $Y$ is included in $M(\textbf{J})$, which violates robust consistency.
\end{proof}

In case we assume the mere consistency of individual judgments, the median property is also not necessary for the consistency of the majority rule. That is, the direction ``if $\Phi_{\MALL}$ violates the median property, then the majority rule is inconsistent'' does not hold. For instance, take the agenda $\{A, \neg A, B, \neg B, A \otimes B, \neg (A \otimes B), C, \neg C\}$. It violates the median property, since it contains a multiset $\{A,  B, \neg (A \otimes B)\}$ which is minimally inconsistent of size 3. However, assuming that each judgment is complete, the majority rule cannot provide an inconsistent outcome on that agenda.
We show this fact as follows.

Since each $J_i$ is complete, either $C$ or $\neg C$ is in $J_i$, therefore either $C$ or $\neg C$ is in $F(\textbf{J})$. We have then two cases: either $F(\textbf{J}) = \{A,  B, \neg (A \otimes B), C\}$ or $\{A,  B, \neg (A \otimes B), \neg C\}$. In either case, $F(\textbf{J})$ is not inconsistent in $\MALL$.\par


 Analogous treatment can be adapted to obtain results for the case of uniform quota rules.
 
\begin{proposition} $\Phi_{\MALL}$ satisfies the kMP iff it is safe for (A, I, N,M), for $F_m$ with $m > n - \frac{n}{k}$.
\end{proposition} 


Since we are assuming that judgments are complete and complement-free, it is meaningful to extend the class of the majority rule, by keeping the axiom of WR and removing one of the other. In this case, we can adapt the known results of Proposition \ref{cl:SSMP} to the case of $\MALL$ under the hypothesis of robust consistency. 
The proof largely adapts the arguments in (\cite{EndrissEtAlAAMAS2010}, Theorem 7, 8 and 9) so we present only a few cases. 

\begin{proposition}\label{prop:mallsmp}
The following facts hold:
\begin{itemize}
\item i)  $\Phi_{\MALL}$ satisfies the SMP iff it is safe for $(WR, A, N ,I)$.
\item ii) $\Phi_{\MALL}$ satisfies the SMP iff it is safe for $(WR, A, N)$.
\item iii) $\Phi_{\MALL}$ satisfies the SSMP iff it is safe for $(WR, A, I)$.
\end{itemize}
\end{proposition}

\begin{proof}
$i$) From left to right, suppose by contradiction that $\Phi_{\MALL}$ satisfies the SMP and $F$ in $[\Phi_{\MALL}, \Phi_{\MALL}](WR, A, N)$ is inconsistent. Therefore there exist two formulas $\phi$ and $\neg \psi$ in $F(\textbf{J})$ such that $\phi$ is equivalent to $\psi$. As we have seen in the proof of Theorem \ref{majority}, since $F$ satisfies $(A)$,$(I)$,$(N)$, the outcome of $F$ only depends on the cardinality of the set of individuals accepting $\phi$ (see also \cite{ListPuppe2009}) and $F$ can be represented by a function $h : \{0,\dots, n=|N|\} \to \{0,1\}$ such that $\phi \in F(\textbf{J})$ iff $h(|N_{\phi}|) = 1$. Since every $J_i$ is robustly consistent, for every $i$, $\phi \in J_i$ iff $\neg \psi \notin J_i$ and since $J_i$ are complete, $\phi \in J_i$ iff $\psi \in J_i$. Thus, by A,I, and N, we can infer that $\phi \in F(\textbf{J})$ iff $\psi \in F(\textbf{J})$. Since by assumption we have $\phi$ in $F(\textbf{J})$, we conclude $\psi \in F(\textbf{J})$, but since $\neg \psi \in F(\textbf{J})$. This violates complement-freeness of $F$. \par
From right to left, assume that $\Phi_{\MALL}$ violates the SMP, we show that there exists a function $F$ in $[\Phi_{\MALL}, \Phi_{\MALL}](WR, A, N ,I)$ that is not robustly consistent. If $\Phi_{\MALL}$ violates the SMP, then there are two formulas $\phi$ and $\psi$ such that $\phi \t \neg \psi$ and $\neg \psi \nvdash \phi$. Define the aggregation procedure $F$ as follows. Take a profile with 3 agents such that agent $J_{1} = \{ \neg A, \neg (A \with B)\}$,  $J_{2} = \{ A, A \with B\}$ and 
$J_{3} = \{A, \neg (A \with B)\}$. Take the aggregator $F$ such that $\phi \in F(\textbf{J})$ if 0 or 1 agents has $\phi \in J_{i}$ and  $\phi \notin F(\textbf{J})$ if more than one agent has $\phi \in J_{i}$. $F$ is clearly satisfies (WR), (A), and (N). Thus, $F(\textbf{J})$ is $\{\neg A, A \with B\}$, which is inconsistent in \MALL.

\noindent
$ii$) The argument above holds also when dropping independence (I). The only difference is that the function $h$ that characterizes the formulas in $F(\textbf{J})$ may vary for each profile $\textbf{J}$.\par
\noindent
$iii$) From left to right, assume that $\Phi_{\MALL}$ satisfies the SSMP and that $F(\textbf{J})$ is inconsistent. Then $F(\textbf{J})$ contains $\phi$ and $\neg \phi$ against complement-freenes of $F$. 
\end{proof}

Again, due to the lack of weakening, if we only assume consistency, the conditions of SMP and SSMP are not necessary for preserving consistency. For instance, the agenda $\{A \with B, \neg (A \with B), \neg A \oplus \neg B, \neg(\neg A \oplus \neg B), C, \neg C\}$ does not satisfy the SMP,  since $A \with B$ is equivalent to $\neg (\neg A \oplus \neg B)$, however we can show that $F$ is consistent on that agenda.\par 

%
%

We conclude this section by noticing that the previous results still applies even if we restrict to the multiplicative fragment of $\MALL$, i.e. they hold for $\MLL$.

\begin{proposition}\label{prop:mll}
$\MLL$ is not safe for MAJ. 
\end{proposition}

This fact can be shown by noticing that the example in the proof of Theorem \ref{MALL} is in fact in \MLL.\par
Again, the median property is adequate to ensure consistency. For example, we have the following result.

\begin{proposition}\label{prop:mpmll}
$\Phi_{\MLL}$ satisfies the median property iff it is safe for MAJ.
\end{proposition}

Moreover, analogous results can be provided for the intuitionistic version of \MALL and \MLL.
 
\subsection{The case of additive linear logic}\label{subsection:all}

We can now state an interesting possibility result for $\ALL$. We assume that each judgment set is also complete. We can show that the majority rule is always robustly consistent on agendas in $\ALL$.
The key property for stating this result is the following.

\begin{property}\label{prop:two}
In additive linear logic (\ALL) every provable sequent contains exactly two formulas (e.g. $A \t B$).
\end{property}
 
This property has been noticed in \cite{HughesGlabbeek2003}.\footnote{This property holds, provided we do not include the logical constants for true and false in the language of \ALL.} If we inspect the additive rules, we see that they cannot add any new proposition. Thus, since every proof starts with axioms $A \t A$, every provable sequent contains two formulas of $\ALL$. This easily entails that there are no minimal inconsistent multisets of size bigger than 2 in \ALL: if $J$ is inconsistent in $\ALL$, then $J \t $ is provable in \ALL. Thus, every agenda in \ALL is safe for MAJ. 

\begin{theorem}\label{ALL} $\ALL$ is safe for MAJ.   
\end{theorem}
  
 \begin{proof}
 Let $\Phi_{\ALL}$ any agenda in $\ALL$. Since in $\ALL$ every provable sequent contains exactly two formulas, if $\Phi_{\ALL}$ contains an inconsistent set $Y$, then $Y$ must contain exactly two propositions, otherwise $Y \t$ would contradict Property \ref{prop:two}. Therefore, any agenda in $\Phi_{\ALL}$ satisfies the median property, thus the majority rule is always robustly consistent. 
 \end{proof} 

\begin{example}\label{ex:all}
For instance, the agenda $\{A, B, A \with B, \neg A, \neg B, \neg (A \with B)\}$ does not violate the median property in $\ALL$. This is due to the fact that
$A, B, \neg(A \with B)$ is not inconsistent in $\ALL$, since $A , B \t A \with B$ is not derivable, hence $A, B, \neg (A \with B) \t$ is not provable.
\end{example}

\noindent
It is worth stressing that Theorem \ref{ALL} is a possibility results for a logic, i.e. $\ALL$ -- that is a set of inference rules -- and not just for a restriction of the language of the agenda, as usual in JA literature (cf. \cite{ListPuppe2009}). This is evident by noticing that $\ALL$ is no longer safe for MAJ, in case we add weakening (W) to the reasoning power of $\ALL$. In presence of weakening, the multiplicative conjunction entails the additive one (cf. \cite{Girard1995}): $A \otimes B \t A\with B$; that is the reason why $\ALL$ plus W is no longer safe, as the following proof shows.

\begin{proposition}\label{ALLW}  The logic \ALL W is not safe for MAJ. 
\end{proposition}

\noindent
\begin{proof} Take the agenda in \ALL that includes $\{A, B, \neg (A \with B)\}$.  
In $\ALL$ W, we have the following proof. 

\begin{center}
$$
\AC{A \t A}
\AC{B \t B}
\rl{$\otimes$R}
\BC{A, B \t A \otimes B}
\AC{A \t A}\rl{$W$ L}
\UC{A, B \t A}
\AC{B \t B}
\rl{$W$ L}
\UC{A, B \t B}
\rl{$\with$R}
\BC{A, B \t A \with B}\rl{$\otimes$L}
\UC{A \otimes B \t A \with B}
\rl{cut}
\BC{A, B \t A \with B}
\rl{$\neg$ L}
\UC{A, B, \neg(A \with B) \t }
\dip
$$
\end{center}
The right branch of the tree is in fact the proof that shows that $A \otimes B$ in entails $A\with B$ in the presence of weakening.  Therefore, there is an agenda in $\ALL$ that is not safe for the MAJ rule evaluated on $\ALL$ W. Thus, $\ALL$ W is not safe for MAJ. 
\end{proof}

\noindent
Proposition \ref{ALLW} does not depend on the language of the agenda, the language is still defined as a subset of $\ALL$. Thus, language restrictions may not be sufficient to guarantee consistency and to analyze the rationality of the outcome of an aggregation procedure. The proof-theoretical analysis is more fine grained and allows for assessing which inference rules are responsible for failures in preserving consistency.
In particular, Propositions \ref{ALLW} shows that there are minimally inconsistent multisets of cardinality 3 when reasoning in $\ALL$ plus W, whereas there are none when reasoning in $\ALL$.\par

The situation of weakening W contrasts with the following case. By adding contraction to $\ALL$, agendas in \ALL C are still safe. The reason is that contraction can only shrink the number of formulas in a provable sequent: If there is no provable sequent with more that two formulas in $\ALL$, the same holds for $\ALL$ C.\par

\begin{proposition}\label{prop:ALLC}
The pair \ALL C is safe for MAJ.
\end{proposition}

In fact, contraction entails that $A \with B \t A \otimes B$, which does not cause any problem with the consistency of the majority rule. 

Since $\ALL$ provides a possibility result for the majority rule, it is worthy to discuss whether this fact extends to larger classes of functions. We show that this is not the case. The following proofs adapt the constructions presented in \cite{EndrissEtAlJAIR2012}. 

\begin{proposition} The following facts hold.
\begin{enumerate}
\item  \ALL is not safe for $(WR, A, N)$
\item \ALL is not safe for $(WR, A,I)$.
\item \ALL is not safe for $(A, N, I, M)$ (quota rules).
\end{enumerate}
\end{proposition}

\begin{proof} 
1. We show that there is an agenda in \ALL and an aggregator that satisfies the axioms above that return an inconsistent outcome. Take the agenda in \ALL that consists of $\{A, \neg A, A \with B, \neg (A \with B)\}$. Consider a profile with 3 agents such that agent $J_{1} = \{ \neg A, \neg (A \with B)\}$,  $J_{2} = \{ A, A \with B\}$ and 
$J_{3} = \{A, \neg (A \with B)\}$. Take the aggregator $F$ such that $\phi \in F(\textbf{J})$ if 0 or 1 agents has $\phi \in J_{i}$ and  $\phi \notin F(\textbf{J})$ if more than one agent has $\phi \in J_{i}$. $F$ is clearly satisfies (WR), (A), and (N). Thus, $F(\textbf{J})$ is $\{\neg A, A \with B\}$ , and $\neg A, A \with B \t $.\par

\noindent
2. Suppose the agenda contains two just two propositions $A$ and $B$ such that $A \t B$ in $\ALL$ and $A \neq \neg B$. Take the constant function $F$ that for every profile accepts $A$ while it rejects $B$. This function is anonymous, independent and weakly rational, however the outcomes violates consistency.\par

\noindent
3. It is enough to set the threshold for acceptance to $q \geq 1$ (i.e. take the union of all formulas that are accepted by at least one agent) to make the 
collective set inconsistent on any agenda in $\ALL$.
\end{proof}

However, by restricting the possible threshold of quota rules to the super-majoritarian quota, a possibility result can be achieved. 

\begin{proposition}\label{quota:ll}
$\ALL C$ is safe for $(A, N, I, M)$, for $F_{m}$ with $m > \frac{n}{2}$
\end{proposition}

\begin{proof}
Since any agenda in \ALL satisfies the median property, inconsistent sets have cardinality 2. A 	quota rule $F_m$ preserves consistency whenever $m > n - \frac{n}{k}$, where $k$ is the size of the largest inconsistent set in \ALL. Thus, $F_{m}$ is consistent for agendas in \ALL if $m > \frac{n}{2}$.
\end{proof}

Therefore, reasoning in $\ALL$ does not provide new possibility results for classes of aggregators that significantly depart from the majority or the super-majority rule. The reason is that consistency in $\ALL$ tightly corresponds to the median property, which is the condition that preserves the majority rule from inconsistency.\par
Once we extend the class of aggregation procedure, the median property is no longer sufficient to guarantee consistency.
In order to guarantee consistency, we need then to apply further restrictions on the agenda in $\ALL$ that strengthen the median property, i.e. the SMP and the SSMP, cf. \cite{EndrissEtAlJAIR2012}.\par

\begin{proposition}\label{prop:allsmp}
The following facts hold:
\begin{itemize}
\item i)  $\Phi_{\ALL}$ satisfies the SMP iff it is safe for $(WR, A, N ,I)$.
\item i) $\Phi_{\ALL}$ satisfies the SMP iff it is safe for $(WR, A, N)$.
\item ii) $\Phi_{\ALL}$ satisfies the SSMP iff it is safe for $(WR, A, I)$.
\end{itemize}
\end{proposition}

We conclude this section with a positive result concerning acceptance-rejection neutrality. We focus on the case in which individual judgements sets are just assumed to be robustly consistent  (i.e. they do not need to be complete). 
\cite{DietrichList2009} have shown that every aggregator that is acceptance-rejection neutral and consistent is a dictatorship of some individual (namely, the aggregator always copies the judgements of some individual). The theorem does not hold for the majority rule, provided we evaluate it w.r.t. to $\ALL$. Therefore, the majority rule is an acceptance-rejection neutral aggregator that preserves robust consistency. 

\begin{proposition}
\ALL is safe for $(arN)$.
\end{proposition}

\section{Judgment Aggregation in Relevant logics}\label{sec:jar}

We discuss now the extension of linear logic to relevant logic. In this case, judgments form multiset again. We assume that the individual judgments are complete. For $\R$, weakening does not hold, therefore the condition of robust consistency is required. The full language of $\R$, that includes multiplicative and additive connectives (i.e. intensional and extensional), behaves as linear logic, thus, we have again the following results.

\begin{theorem}
$\R$ is not safe for MAJ.
\end{theorem}

\begin{proof}
It is sufficient to use the argument in the proof of Theorem \ref{MALL}.
\end{proof}

Moreover, the median property provides again a suitable condition for the consistency of the majority rule. 

\begin{proposition} 
$\Phi_{\R}$ satisfies the median property iff it is safe for MAJ.
\end{proposition}

General possibility results can be obtained by restricting to $\AR$, the additive fragment of relevant logic. We extend the possibility result for $\ALL$ to the additive fragment of $\R$ by proving it by means of the Hilbert system. 

\begin{lemma}\label{lemma:C}
In $\HMALL$, the axiom for contraction (C): $\t (\phi \multimap (\phi \multimap \psi)) \multimap (\phi \multimap \psi)$ is derivable iff the following rule of contraction (HC) is:

$$
\AC{\{\Gamma, \phi,\phi\} \t \psi}
\rl{HC}
\UC{\{\Gamma, \phi\}\t \psi}
\dip
$$
\end{lemma}

\begin{proof}
Assume the axiom $C$ holds. We derive the rule HC. Assume $\{\phi,\phi\} \t \psi$. By the deduction theorem, we have that $\t \phi \multimap (\phi \multimap \psi)$ is derivable. 
By the axiom (C) and the $\multimap$-rule, we can derive:

\begin{itemize}
\item $\t \phi \multimap (\phi \multimap \psi)$ (hypothesis)
\item $\t \phi \multimap (\phi \multimap \psi) \multimap (\phi \multimap \psi)$ (axiom (C))
\item $\t \phi \multimap \psi$ ($\multimap$-rule)
\end{itemize}

Thus, by the deduction theorem, we have that $\{\phi\} \t \psi$.\par
\noindent
In the other direction, \HMALL plus the rule (HC) derives (C). We show that the axiom (C) can be derived as follows. In $\HMALL$, it is derivable that:
$\{\phi, \phi, \phi \multimap (\phi \multimap \psi)\} \t \psi$. By applying (HC), we obtain that $\{\phi, \phi \multimap (\phi \multimap \psi)\} \t \psi$, which by the deduction theorem entails 
$\t (\phi \multimap (\phi \multimap \psi)) \multimap (\phi \multimap \psi)$.
\end{proof}

Moreover, if we assume axioms (D1) and (D2), then by the deduction theorem, we have that:

\begin{itemize}
\item (HD1): $\{A \with (B \oplus C)\} \t (A \with B) \oplus (A \with C)$
\item (HD2)$\{ (A \with B) \oplus (A \with C)\} \t  A \with (B \oplus C)$\\
\end{itemize}

We can now extend Property \ref{prop:two} to \AR as follows.

\begin{property}\label{two:CD}
In \AR every inconsistent set has cardinality 2. 
\end{property}

\begin{proof}
We know that in \ALL, every provable sequent has cardinality two. Since the sequent calculus \MALL and its Hilbert systems are equivalent, that entails that if
$\Gamma \t \bot$ in \HMALL, then $\Gamma$ contains exactly at most two formulas. The Hilbert system for $\R$ is obtained by adding to \HMALL the rule (HC), (HD1) and (HD2). It is easy to see that none of this can increase the number of formulas in a derivation.  
\end{proof}

We can now show that the additive fragment of \R is safe for MAJ. 

\begin{theorem}\label{th:ar}
$\AR$ is safe for MAJ.
\end{theorem}

\begin{proof}
By Lemma \ref{lemma:C} and Property \ref{two:CD}, every minimally inconsistent set in \AR has cardinality 2. Thus, every agenda in \AR is safe for MAJ. 
\end{proof}

The other safety results and the characterisation of safe agendas can be easily restated for the case of $\R$.  Many systems of relevant logics have been discussed in the literature, cf. \cite{Anderson1992entailment,Restall2002}. A dedicated treatment of judgment aggregation in relevant logics is left for future work.

\section{Extensions to other logics}\label{sec:extensions}

We have approached JA in substructural logics as they allows for covering a significant number of non-classical logics. As we discussed, any extension that allows for weakening is going to suffer of the analogous problems of aggregation of classical logic. It is however possible to look for extensions of the reasoning power of \ALL and \AR ---for which safety is ensured--- to other logics that lack weakening. An interesting minimal class of \emph{fuzzy logics} can be defined as an extensions of $\MALL$. The class of \emph{uninorm fuzzy logics} is obtained from \HMALL by adding, besides distributivity, the \emph{linearity} axiom, which can be written in our notation, by means of the axiom (L): $(A \multimap B) \oplus (B \multimap A)$ (cf. \cite{metcalfe2008proof}, p. 54).  As the axiom does not increase the number of formula in a derivation, it is in principle possible to extend our treatment and the safety results of the additive fragment to uninorm fuzzy logics.\par
Another reasonable extension is to discuss the case of the logic of \emph{bunched implication} \BI \cite{pym2013semantics} which have significant applications in computer science. From an axiomatic perspective, this logic is close to the relevant logic \R ---since its additive conjunction and disjunction are distributive--- (cf \cite{Paoli2002}, p. 130).  Moreover \BI blocks weakening at the general level and it separates additive and multiplicative connectives. Thus, it appears a good candidate for extending the safety results of \ALL. However, \BI combines an \emph{intuitionistic} implication and a substructural implication. That is, the additive implication of \BI is more poweful than the additive implication of \LL. In particular, the additive implication of \BI is in fact an intuitionistic implication that is related, by the adjunction $A \wedge B \leadsto C$ iff $A \leadsto (B \leadsto C)$, to the additive conjunction. This causes Property \ref{prop:two} to fail for the additive fragment of \BI, since $A, A \leadsto B, \neg B$ is inconsistent in \BI. Therefore, \BI is not in general safe for the axioms that characterise the majority rule. 
Adding a strong negation to the intuitionistic fragments that we have explored is also a viable extension to justify the completeness assumption concerning the individual judgments. The axioms that capture a strong constructive negation --- see  \cite{Wansing2007}--- preserve Proposition \ref{prop:two}, thus in principle adding a strong negation provides a viable extension to intuitionistic \ALL and \AR.\par

\subsection{Two-logics judgment aggregation}

The modelling that we have introduced enables the definition of aggregators that link profiles of judgment expressed in a certain logic to collective outcomes expressed in a possibly different logic.  In particular, our previous results can be extended to aggregation procedures that takes profiles of judgments that are defined on agendas in \emph{classical} logic and return judgments in \ALL or \AR. Firstly, note that classical logic is obtained, for instance from \ALL, by adding weakening and contraction. Thus, any agenda in $\Phi_{\ALL}$ can faithfully represent the deductive structure of a classical agenda once we reason on its formulas of by means of  \ALL plus contraction and weakening (\ALL C W) . By adding weakening and contraction, the multiset of formulas of the agenda in \ALL is deductively equivalent to a set, since multiple occurrences of a formula can be identified by means of contraction (C).\par
By inspecting the proofs of Theorem \ref{ALL} and \ref{th:ar}, one can see that the proof still holds in case individual reasoning in \ALL or \AR is extended to individual reasoning in \CL. The reason is that, by reasoning in \CL more theorems (more sequents) can be derived, thus, in principle, there are more provably inconsistent sets. However, individual judgements are assumed to be robustly consistent ---that in this case coincides with standard consistency because of weakening and contraction--- therefore, although more (individually) inconsistent sets may be derived in classical logic, they are all prevented by the (robust) consistency assumption. For those reasons, we obtain safety results for MAJ and pairs of agendas $(\Phi_{\CL}, \Phi_{\ALL})$ and $(\Phi_{\CL}, \Phi_{\AR})$, which entails the following result.\footnote{In this case, the aggregation procedures basically associate profiles of sets to mutlisets, this is done by viewing the set that would be the output of the aggregation procedure as multiset with multiplicity one}

\begin{corollary}(of Theorem 29 and 40). The pairs of logics $(\CL, \ALL)$ and $(\CL, \AR)$ are safe for MAJ. 

\end{corollary}

Recall that a pair of logics is safe for a certain class of axioms if for every agenda defined on the fist logic, every aggregation functions that satisfies those axioms returns sets of judgments that are robustly consistent when assessed with respect to the second logic (cf. Definition 10).\par
By assuming that individuals reason by means of classical logic, we are meeting the assumptions of the standard model of JA and of the impossibility result in \cite{ListPettit2002}. Therefore, by assessing collective rationality in a logic that is significantly weaker that classical logic, i.e. \ALL or \AR, we can actually circumvent the classical impossibility result of \cite{ListPettit2002}. By assessing collective rationality in \ALL or \AR, we can guarantee the preservation of a consistency property from the individual judgments to the collective judgments: if the individual judgments are consistent w.r.t. classical logic, then the majority rule provides judgments that are consistent w.r.t. additive linear or relevant logic.\par 
This strategy for circumventing impossibility results maintains all the normative axioms of the majority rule, it only weaken the notion of rationality that is suitable for collective entities.\par
One may wander what happens if we weaken the logic that models individual reasoning while keeping a strong logic that assesses the collective outcome, for instance if individuals reason by means of \ALL and the outcome of the majoritarian aggregation is assessed by means of \CL.  In this case, the preservation of (robust) consistency is lost. In particular, suppose that  $\{A, B, \neg (A \with B)\}$ is the outcome of the majority rule: it is not inconsistent in \ALL, but it is in \CL, once we enable also weakening (W).


\section{Related work}\label{sec:relatedwork}

A number of articles that investigate logics for judgment aggregation are related to this paper. Proof-theoretical methods in social choice theory and judgment aggregation have been previously used to provide formal proofs of important theorems such as Arrow's theorem. Hilbert system presentations have been used in \cite{aagotnes2011logic} and \cite{CinaEndrissAAMAS2015}, and natural deduction for modal logics of JA has been introduced in \cite{Perkov2016}. The use of proof-theory ---and in particular of sequent calculi--- for modeling collectively accepted beliefs has been developed by \cite{hakli2011reasoning}, although the logic that has been introduced there admits weakening. The use of sequent calculi for modelling inferences in judgment aggregation has been introduced in \cite{PorelloIJCAI2013,PorelloECAI2012}.\par
Logics for judgment aggregation have been mainly developed as modal logics that extend classical propositional reasoning, e.g. \cite{Pauly2007} uses non-normal modal logics to axiomatize the majority rule and \cite{aagotnes2011logic} and \cite{CinaEndrissAAMAS2015} use (normal) modal logics to prove theorems in judgment aggregation. However, a small number of articles discussed non-classical logics related to judgment aggregation. G\"odel-Dummett Logic has been used to compare judgment and preference aggregation by \cite{grossi2009unifying}. G\"odel-Dummett logic is an extension of intuitionistic logic, therefore weakening holds; from the perspective of the safety of a logic for a set of axioms, this case is comparable to the case of intuitionistic logic. Moreover, recent work discusses judgment aggregation within the framework of argumentation theory, that is based on default logics, see for instance \cite{caminada2011judgment}.\par
Finally, a recent work that extends the general model of judgment aggregation of \cite{Dietrich2007} to the case of nonmonotonic logics is \cite{WenJANMSYnthese2017}. The impossibility results of \cite{Dietrich2007} are there extended to the case of nonmonotonic logic. Here, we have approached non-monotonicity from the perspective of substructural logics, that is, we have discussed logics that lack weakening. Notice that, as we have argued, the lack of weakening is not sufficient for the possibility results of Theorem 29 and 40. The distinction between multiplicatives and additives (or intensionals and extensionals) and the restriction to the additives is required. Multiplicative linear logic lacks weakening, however it does not guarantee the consistency of the majority rule. A comparison between the treatment of non-monotonicity in substructural logics (i.e. the lack of weakening) and the non-monotonicity of default logics and of nonmonotonic logics requires a dedicated treatment and it is left for future work.\par
Another example of non-classical logic ---or more precisely of non-classical definition of the conditional--- in judgment aggregation is the approach based on subjunctive implications of Lewis's conditional logic by \cite{dietrich2010possibility}. The motivation for studying subjunctive conditional is close to the motivations of this paper --- and of a number of non-classical logics--- that is, that of exploring meanings of conditionals that are close to the meaning of implications used by human agents and that do not suffer of the paradoxes of classical material implication. For subjunctive implications, Dietrich shows possibility results for quota rules. Moreover, the notion of relevant premises in a judgment aggregation problem has been discussed in \cite{dietrich2015aggregation} by tuning the axioms of aggregation procedures. Here, we have approached the problem of relevance by studying judgement aggregation within relevant and linear logics, that is, by introducing suitable logical operators.\par

An important group of articles that is related to the analysis of this paper is the semantic approach to logics for judgment aggregation due to \cite{Dietrich2007}. In particular, \cite{Dietrich2007}, \cite{dietrich2010premiss} and \cite{mongin2012doctrinal} discuss JA with respect to a ``general logic'', that includes a number of significant logics, providing an important step towards generalizing the model of judgment aggregation. The methodology they use is algebraic: a logic is defined by means of a (semantic) \emph{consequence relation} that satisfies a number of desiderata that are met by a variety of logics (e.g. by classical logic and by many modal logics). In particular, the consequence relation is there assumed to be monotonic, which is equivalent to assuming weakening from a proof-theoretical perspective. The results of \cite{Dietrich2007} provide a general impossibility for judgment aggregation for logics that are monotonic (besides the other conditions). Moreover, a similar algebraic approach for nonmonotonic logic has been developed in \cite{WenJANMSYnthese2017}. The case of linear and relevant logics still escape a straightforward reformulation of the approach by  \cite{Dietrich2007} due to the distinction between multiplicatives and additives. A close comparison with the semantic approach proposed by Dietrich and in \cite{WenJANMSYnthese2017} requires a dedicated future work.\par
Another interesting algebraic approach to discussing JA in a variety of logics is given by \cite{herzberg2013universal} and \cite{esteban2015abstract}. In that case, the treatment applies to a wide class of distributive logics, thus our sections on linear logic may provide an extension of their analysis, once rephrased in their algebraic counterparts.

\section{Conclusions and future work}\label{sec:conclusion}

We have seen an overview of results for judgment aggregation in a number of important non-classical logics. For monotonic logics ---that are here understood as logics that satisfy weakening, e.g. for intuitionistic logic and for extensions of classical logic--- we have seen that a simple rephrasing of the usual modelling of JA is sufficient to establish the expected results. For the case of non-monotonic logics ---e.g. Lambek calculus, linear logic, and relevant logics--- the modelling requires a careful examination of the consistent sets of formulas and of the operational meaning of logical connectives, in order to adapt the standard characterisations of safe agendas (e.g. the median property).\par
By studying JA in weak logics, and by means of a proof-theoretical approach, a fine-grained analysis of the inferences that are responsible of collective inconsistency can be provided. We have seen that by defining agendas in the additive fragment of linear or relevant logic, general safety results are viable. However, adding weakening is sufficient for collective inconsistency to stike back. By contrast, contraction does not affect safety results.\par 
We have also seen that even if individuals reason by  means of the full power of classical logic ---as in the standard JA setting--- it is nonetheless possible to asses collective rationality with respect to the additive fragment of linear or relevant logic, in order to preserve consistency. 
Understanding collective rationality with respect to a weak logic provides then a way to circumvent the impossibility theorem provided by \cite{ListPettit2002}, at least in the sense that discursive dilemmas are no longer construed as a problem of logical consistency. Substructural reasoning provides a rationalization of classical collective inconsistencies, thus the problem of ascribing irrational attitudes to collective entities ---that threaten an aggregative view of collective agents, see for instance \cite{ListPettit2011}--- is prevented by the use of substructural reasoning. A detailed analysis of discursive dilemmas from the perspective of substructural logics in left for future work.\par
We did not discuss here the computational complexity of the decision problems involved in our setting. For the systems that we have introduced, the computational complexity of theorem proving is known, therefore, it is possible in principle to adapt the treatment of \cite{EndrissEtAlJAIR2012} to the non-classical logics. For instance, multiplicative additive linear logic is known to be PSPACE-complete. From that, one can approach the computational complexity of checking whether an agenda of propositions in \MALL has the median property by adapting the arguments in \cite{EndrissEtAlJAIR2012} (e.g. Lemma 20).\par 
We have covered a representative number of significant non-classical logics and of non-classical connectives that can be in principle applied to model specific aggregation problems. For instance, we have introduced an order-sensitive implication that may capture dependencies between propositions in the agenda, a resource-sensitive implication that captures a number of aspects of causality and resource sensitivity, and we presented relevant implication that captures the connection between the antecedent and the consequent of true implications.
We have only marginally discussed a number of logical systems that can be obtained as extensions of linear or relevant logic. The case of many-valued and fuzzy logics and the case of paraconsistent logics have been discussed here only within the substructural perspective. An extension of the present treatment to further logical calculi is the matter of future work. 

\section*{Acknowledgments}

This work extends \cite{PorelloIJCAI2013} and \cite{PorelloECSI2014}. I thank the anonymous reviewers for the insightful comments.

\bibliographystyle{apa}
\bibliography{jancl}
\end{document}